\documentclass{article}

\usepackage[margin=1in]{geometry}
\usepackage{amsmath, amsthm, amssymb, graphicx}
\usepackage[dvipsnames]{xcolor}
\usepackage{url}

\DeclareMathOperator*{\essinf}{ess inf}
\DeclareMathOperator*{\esssup}{ess sup}
\DeclareMathOperator*{\argmin}{arg min}
\DeclareMathOperator*{\argmax}{arg max}

\def\l{\left}
\def\r{\right}
\def\p{\partial}

\newcommand{\R}{\mathbb{R}}
\renewcommand{\P}{\mathbb{P}}
\newcommand{\Q}{\mathbb{Q}}
\newcommand{\E}{\mathbb{E}}
\newcommand{\var}{\mathbb{V}ar}
\newcommand{\F}[1]{\mathcal{F}_{#1}}
\newcommand{\C}[2]{\mathcal{C}_{#1,#2}}
\newcommand{\ofp}{(\Omega,\mathcal{F},  \mathbb{P})}
\newcommand{\offp}{(\Omega,\mathcal{F}, \mathbb{F}, \mathbb{P})}
\newcommand{\ofgp}{(\Omega,\mathcal{F}, \mathbb{G}, \mathbb{P})}
\newcommand{\lofsp}{L^\infty(\Omega,\mathcal{F}_s,  \mathbb{P})}
\newcommand{\lofssp}{L^\infty(\Omega,\mathcal{F}_\sigma,  \mathbb{P})}
\newcommand{\loftp}{L^\infty(\Omega,\mathcal{F}_t,  \mathbb{P})}
\newcommand{\lofttp}{L^\infty(\Omega,\mathcal{F}_\tau,  \mathbb{P})}
\newcommand{\lofp}{L^\infty(\Omega,\mathcal{F},  \mathbb{P})}
\def\hatrho{\hat{\rho}}
\def\taustar{\tau^{*}}
\def\ind{{\mathchoice{1\mskip-4mu\mathrm l}{1\mskip-4mu\mathrm l}
		{1\mskip-4.5mu\mathrm l}{1\mskip-5mu\mathrm l}}}
\newcommand{\T}[2]{\mathcal{T}_{#1,#2}}
\newcommand{\D}[1]{\mathcal{D}_{#1}}

\usepackage[normalem]{ulem}

\newcommand{\blue}[1]{{ \color{blue}{{#1}}}}

\newcommand{\xred}[1]{{\color{red}{{\sout{{#1}}}}}}
\newcommand{\xxred}[1]{{\color{red}{{\xout{{#1}}}}}}

\newtheorem{definition}{Definition}[section]
\newtheorem{proposition}[definition]{Proposition}
\newtheorem{remark}[definition]{Remark}

\newtheorem{lemma}[definition]{Lemma}
\newtheorem{theorem}[definition]{Theorem}
\newtheorem{assumption}[definition]{Assumption}

\def\bRm{\bar{R}^m}
\def\Zm{\bar{\bf Z}^m}
\def\W{{\bf W}}

\def\hm{\bar{h}^m}
\def\bZ{{\bf Z}}
\def\HR{\hat{R}}
\def\HZ{\hat{Z}}
\def\HK{\hat{K}}
\def\hm{\bar{h}^m}

\title{Risk-indifference Pricing of American-style Contingent Claims}

\author{
Rohini Kumar\thanks{%
R. Kumar, Department of Mathematics, Wayne State University, 656 W. Kirby Street, Detroit, MI 48202, USA (e-mail: 
\texttt{rohini.kumar@wayne.edu}); Corresponding Author.}
\and
Frederick ``Forrest" Miller\thanks{%
F. Miller, Department of Mathematics, Northeastern University, Lake Hall, 360 Huntington Avenue, Boston MA 02115, USA (e-mail: 
\texttt{miller.f@northeastern.edu}).}
\and
Hussein Nasralah\thanks{%
H. Nasralah, Department of Mathematics and Statistics, University of Michigan--Dearborn, 4901 Evergreen Road, Dearborn, MI 48128, USA (e-mail: 
\texttt{hnasrala@umich.edu}).}
\and
Stephan Sturm\thanks{ %
S. Sturm, Department of Mathematical Sciences, Worcester Polytechnic Institute, 100 Institute Road, Worcester MA 06109, USA
	(e-mail: \texttt{ssturm@wpi.edu}).}
}

\begin{document}

\maketitle

\begin{abstract}
This paper studies the pricing of contingent claims of American style, using indifference pricing by fully dynamic convex risk measures. We provide a general definition of risk-indifference prices for buyers and sellers in continuous time, in a setting where buyer and seller have potentially different information, and show that these definitions are consistent with no-arbitrage principles. Specifying to stochastic volatility models, we characterize indifference prices via solutions of Backward Stochastic Differential Equations reflected at Backward Stochastic Differential Equations and show that this characterization provides a basis for the implementation of numerical methods using deep learning.
\end{abstract}

\vspace{5mm}
 
\begin{flushleft}
	 \textbf{Keywords:} American Options, Fully Dynamic Convex Risk Measures, Indifference Pricing, (Reflected) Backward Stochastic Differential Equations.\\
	 \textbf{Mathematics Subject Classification (2020):} 91G20, 91G70, 60H10.\\
	 \textbf{JEL classification:}  D81, G13,  C61.
\end{flushleft}

\section{Introduction}

The pricing of American style derivatives remains an active area of research. Most single stock options are of American style, thus driving the need to understand American claims. American options also appear implicitly in other domains of financial research, e.g., as real options, i.e., for the valuation of capital investments using option-pricing methods. The absence of a closed-form benchmark (as the Black--Scholes formula in the European case) channels their research into numerical methods. Additionally,  theoretical questions on optimal stopping in nonlinear market models saw recent progress, in particular in incomplete markets.

While the no-arbitrage principle guarantees a unique price in complete markets, in incomplete markets it provides only price bounds (super- and sub-hedging prices) that are typically very wide and do not provide a practical indication of a reasonable price. Therefore, further techniques have been developed to characterize fair and reasonable prices. One of the most prominent ones is indifferences pricing, developed first by Hodges and Neuberger \cite{HN89}; see the book \cite{Car09} (edited by Carmona) for a survey. The goal is to establish a threshold or reservation price, at which a potential buyer is indifferent between buying the claim for this price or not buying it, while in either case allowing for continuous trading in the underlying market. Originally developed in a framework of utility maximization, it has been extended to other criteria, such as forward performance measures \cite{LSZ12} and risk measures \cite{BEK09, KS05}; it has also be used to construct horizon-independent risk measures \cite{ZZ10, CHLZ19}. The formulation via convex dynamic risk measures is in particular attractive, as it is not only nicely connected to the theory of Backward Stochastic Differential Equations (BSDEs), but relies on concepts widely used in the industry and in line with the current regulatory framework \cite[Section 12]{Hul18}.

Monetary risk measures were first introduced by Artzner, Delbaen, Eber and Heath \cite{ADEH99} in the form of coherent measures of risk and then generalized to convex risk measures by F\"{o}llmer and Schied \cite{FS02} and Fritelli and Rosazza Gianin \cite{FRG05}. To capture the time evolution of risk, conditional and dynamic versions have been developed (see, e.g., Cheridito, Delbean and Kupper \cite{CDK04}), and a close connection to Backward Stochastic Differential Equations has been established, see, e.g.,  Peng \cite{Pen04} and Rosazza Gianin \cite{RG06}. Most recently fully dynamic convex risk measures became a focus point, i.e., risk measures where both time parameters, the horizon and the evaluation time are considered dynamic, see \cite{BNDN20} and \cite{RGDN24}. A crucial property in this context is time consistency; it has been studied extensively, see, e.g., the overview papers by Acciaio and Penner \cite{AP11} and Bielecki, Cialenco and Pitera \cite{BCP17} in discrete time and Rosazza Gianin and DiNunno \cite{RGDN24} for fully time-consistent risk measures in continuous time.

The use for indifference pricing was pioneered by Xu \cite{Xu06} and Barrieu and El Karoui \cite{BEK04}, generalizing earlier results by Rouge and El Karoui \cite{REK00} for exponential utility, which can be reinterpreted as an entropic risk measure, and studied systematically in a general setting by Kl\"{o}ppel and Schweizer \cite{KS05} and Barrieu and El Karoui \cite{BEK09}. Applications to stochastic volatility models and the inverse problem of calibrating risk measures to market data has been studied in Sircar and Sturm \cite{SS15} and Kumar \cite{Kum15}, see also \cite{ES10}.

Indifference pricing for American options appeared first in the study of transaction costs by Davis and Zariphopoulou \cite{DZ95}, extending work by Davis, Panas and Zariphpoulou \cite{DPZ93} for the European case. The literature encompasses \cite{BZ14, Dam06, GK22, Kue02, LS09, LSZ12, OZ03, WD09, YLY15, Zak05} who use utility functions, stochastic differential utilities and forward performance processes as criteria for indifference, we are not aware of any use of dynamic risk measures for the American case (despite the use of Reflected Backward Stochastic Differential Equations for American options dating back to El Karoui, Pardoux and Quenez \cite{EKPQ97}, based on theoretical results in \cite{EKKPPQ97}). For a full discussion of the papers of indifference pricing of American options, we refer to Section \ref{sec:lit-comp}.

The conceptual challenges in implementing indifference pricing stem from the fact that option buyer and seller share different perspectives and the seller's pricing consideration has to take into account the buyer's exercise decision. This calls for a careful consideration of which strategies are admissible - something that has been carefully studied in the case of finitely many payoff options by K\"{u}hn \cite{Kue02} and whose perspective we amend with a general counterpart. The pricing by the buyer is slightly more straightforward, however one has to carefully consider at which time one imposes indifference (at the exercise time or maturity?) and how this connects to the notion of arbitrage. In our opinion clarity is best achieved when considering a general case, where one allows buyer and seller to work in different filtrations, reflecting difference in access to market information.

We then specialize to the setting of stochastic volatility models, following the general setup of \cite{SS15} and \cite{Kum15}. We find that the American indifference prices can be described through Backward Stochastic Differential Equations reflected at Backward Stochastic Differential equations (BSDE-R-BSDEs for short), i.e., Reflected Backward Stochastic Differential Equations (RBSDEs) for which the reflection boundary is given by a BSDE itself. This structure reflects that risk mitigation through trading in the market continues after the exercise of the option -- we observe that the reflecting boundary encapsulates, in addition to the exercise value of the option,  the risk of holding a zero contract from exercise time to maturity (cf. Remark \ref{rem:bdry}). Also, in this setting, the proof of the characterization of the seller's price requires a substantial amount of work, while the characterization of the buyer's price is more straightforward. We illustrate our findings by means of a numerical example, the pricing of an American put option, for which we use deep learning methods to simulate the BSDE-R-BSDEs.

New methods on solving BSDEs and RBSDEs through deep learning have been proposed recently and show much promise. Initially, E, Han and Jentzen developed the DEEP BSDE Solver as a forward scheme \cite{EHJ17, EHJ18}, able to tackle high dimensional problems. A global backward scheme, the Backward Deep BSDE Method was developed by Wang, Chen, Sudjianto, Liu and Shen \cite{WCSLS18} and studied in detail by Gao, Gao, Hui and Zhu \cite{GGHZ23}, who also analyze the convergence in the case of Lipschitz drivers and show how to use the scheme for Bermudan options. Hur\'{e}, Pham and Warin \cite{HPW20} introduced schemes based on dynamic programming, namely the Deep Backward Dynamic Programming schemes, containing in particular one for RBSDEs on which our simulations rely. Recent overview articles on work in this general direction can be found in the papers by E, Jentzen and Han \cite{EJH22} as well as Chessari, Kawai, Shinozaki and Yamada \cite{CKSY23}.

The paper is structured as follows. Section \ref{sec:pricing} provides a general setup for the risk-indifference pricing of claims of American style, with minimal conditions on the risk measures involved and considering potentially different information available to buyers and sellers. Section \ref{sec:stochvol} studies the risk-indifference pricing of American claims in stochastic volatility models via BSDE-R-BSDEs, and Section \ref{sec:deep} provides a numerical implementation via deep learning. Section \ref{sec:conc} concludes by reviewing the contributions of the current work.

\section{Risk-indifference pricing}\label{sec:pricing}

To elucidate the conceptual ideas at the heart of our problem, we initially consider general risk measures and American contingent claims in a general semimartingale setup, in which the buyer and seller have access to different information represented by different filtrations. We first explain the market setup and the notion of fully dynamic risk measures, before defining the indifference prices from a buyer's and seller's perspective and showing that they are free of arbitrage. We conclude this general section by reviewing our methodology against the backdrop of the existing literature on indifference pricing of American claims.

\subsection{Setup}

We consider a filtered probability space $\ofgp$ with complete and right-continuous filtration $\mathbb{G} = (\mathcal{G}_t)_{0 \leq t \leq T}$ assuming $\mathcal{G}_T = \mathcal{F}$. We consider a financial market consisting of risky assets modeled by a $d$-dimensional $\mathbb{G}$-continuous semimartingale $\hat{S}$ and a riskless asset modeled by a continuous, non-decreasing $\mathbb{G}$-adapted process $B$.

To exclude arbitrage in the sense of \textit{no free lunch with vanishing risk} (see \cite{DS06}), we require the existence of a probability measure $\Q$ under which the discounted asset process $S=\hat{S}B^{-1}$ is a local $\mathbb{G}$-martingale. The information of the buyer of the option is given by the complete and right-continuous filtration $\mathbb{F}^\text{buy}$, while that of the seller is $\mathbb{F}^\text{sell}$;   $\mathbb{F}^{S,B} \subseteq \mathbb{F}^\text{sell}, \mathbb{F}^\text{buy} \subseteq \mathbb{G}$, where $\mathbb{F}^{S,B}$ denotes the (augmented) natural filtration generated by risky and riskless assets. For general statements that do not require one particular filtration, we will use the generic $\mathbb{F}$. This setup allows us to model situations where the seller and buyer rely on additional, potentially differential, private information (and randomization), while precluding arbitrage opportunities.

Trading in the market is continuous and self-financing. A portfolio $\hat{V}$ at time $t$ is given by holding $h_t$ shares of $\hat{S}$ and $\eta_t$ shares of $B$, hence
\[
	\hat{V}_t = h_t \hat{S}_t + \eta_t B_t
\]
for all $t \in [0,T]$ where $h, \eta \in \tilde{\mathcal{H}}_\text{buy/sell}$, the set of $\mathbb{F}^{\text{buy/sell}}$-predictable strategies, for the buyer (resp. seller). For the discounted portfolio dynamics $V = \hat{V}B^{-1}$ we have therefore (starting with zero initial wealth), $V_t = \int_0^t h_s \, dS_s$ and use $V_{s,t} = V_t -V_s$ to denote portfolios coming from hedging from time $s$ on. To mark the dependence of the portfolio on the hedging strategy $h$ we will use a superscript, writing $V^h$. To avoid doubling strategies, we will restrict ourselves to strategies with bounded wealth. The set of bounded claims hedgeable at no cost from from time $s$ to $t$ is
\[
\mathcal{C}^\text{buy/sell}_{s,t} := \biggl\{V^h_{s,t} = \int_s^t h_u \, dS_u \, : \, \bigl\vert V_{s,u}^h \bigr\vert  \leq M \text{ for some constant } M \in \mathbb{R}, h \in \tilde{\mathcal{H}}_\text{buy/sell} \text{ and all } u \in (s,t]			\biggr\} .
\]
The strategies leading to these claims we denote by $\mathcal{H}_\text{buy/sell} := \bigl\{ h \in \tilde{\mathcal{H}}_\text{buy/sell} \, : \, V^h_{0,T} \in \mathcal{C}^\text{buy/sell}_{0,T}\bigr\}$.

The goal of the paper is to determine a price for an American style contingent claim $\hat{\xi}$, i.e., an almost surely continuous, bounded and $\mathbb{F}^{S,B}$-adapted process that can be exercised by the buyer at an $\mathbb{F}^\text{buy}$-stopping time $\tau$, paying $\hat{\xi}_\tau$ where $\tau$ is a stopping time either on $[0,T]$ (American style claim) or a closed countable subset of it (Bermudan style claim). We denote the set of all $\mathbb{G}$-stopping times larger than or equal to $t$ by $\mathcal{T}_{t,T}$ and write $\mathcal{T}_{t,T}^{\text{buy/sell}}$ for the stopping times measurable with respect to the buyer's (resp. seller's) filtration. In line with the notation introduced above we write $\xi = \hat{\xi}B^{-1}$ for the discounted claim.

The absence of arbitrage in the financial market consisting of $\hat{S}$ and $B$ guarantees the existence of an arbitrage free price for the derivative $\hat{\xi}$. However, unless the market is complete (i.e., the local martingale measure $\Q$ is unique), the no-arbitrage principle does not provide a unique price, but rather a (practically often very large) interval of arbitrage free prices.

The current article is concerned with the choice of a \textit{reasonable} price among the multitude of arbitrage free prices. Our pricing method is based on the principle of indifference, i.e., we determine the price for which the buyer (resp. seller) is indifferent between buying the option for this price (and hedging in the underlying market) or not buying the option at all (but still hedging in the market).

The indifference criterion we will use is that of indifference in risk. For that purpose we introduce the notion of fully dynamic convex risk measures (see \cite{BNDN20}):
	
\begin{definition}
	A family of mappings $\rho_{s, t}: \loftp \to \lofsp$,  with times $s, \,  t$ satisfying $0 \leq s \leq  t \leq T$,  is called a (strongly) time-consistent fully dynamic convex risk measure if it satisfies the following properties:
	\begin{itemize}
		\item[A)] \textbf{Monotonicity}: For all $\xi_1, \, \xi_2 \in \loftp$, $\xi_1 \geq \xi_2$ $\mathbb{P}$-a.s., and for all $s \leq t$,
		\[
			\rho_{s,t} (\xi_1) \leq \rho_{s,t} (\xi_2) \quad \mathbb{P}-a.s.
		\]
		\item[B)] \textbf{Cash-Invariance}: For all $\xi \in \loftp$, $m_s \in \lofsp$, and for all $s \leq t$,
		\[
			\rho_{s,t} (\xi + m_s) = \rho_{s,t} (\xi) - m_s \quad \mathbb{P}-a.s.
		\]
		\item[C)] \textbf{Convexity}: For all $\xi_1, \, \xi_2 \in \loftp$, $\lambda \in [0,1]$, and for all  $s \leq t$,
		\[
			\rho_{s,t} \bigl(\lambda \xi_1 + (1-\lambda) \xi_2\bigr) \leq \lambda \rho_{s,t} (\xi_1) + (1-\lambda) \rho_{s,t} (\xi_2)  \quad \mathbb{P}-a.s.
		\]
		\item[D)] \textbf{Time-consistency}: For all $\xi_1, \, \xi_2 \in L^\infty(\Omega,\mathcal{F}_u,  \mathbb{P})$ and $s \leq t \leq u$,
		\[
			\rho_{t,u} \bigl(\xi_1\bigr) = \rho_{t,u} \bigl(\xi_2\bigr) \quad \Longrightarrow \quad \rho_{s,u} \bigl(\xi_1\bigr) = \rho_{s,u} \bigl(\xi_2\bigr)  \quad \mathbb{P}-a.s.
		\]
	\end{itemize}
	There is also a stronger form of time consistency,
	\begin{itemize}
		\item[D')] \textbf{Strong time-consistency}: For all $\xi \in L^\infty(\Omega,\mathcal{F}_u,  \mathbb{P})$ and $s \leq t \leq u$,
		\[
			\rho_{s,t} \bigl( -\rho_{t,u} (\xi)\bigr) = \rho_{s,u} (\xi)  \quad \mathbb{P}-a.s.
		\]
		This property clearly implies D), so it is a stronger assumption.
	\end{itemize}
\end{definition}

An important additional property that follows from this definition (see \cite[Section 3]{KS05}) is 
\begin{itemize}
	\item[F)] \textbf{$\mathcal{F}_t$-regularity}: For all $\xi_1, \xi_2 \in \loftp$, $A \in \mathcal{F}_s$ and $s \leq t$,
	\[
		\rho_{s,t} \bigl(\xi_1 \ind_A + \xi_2 \ind_{A^c}\bigr) = \rho_{s,t} (\xi_1\bigr)\ind_A + \rho_{s,t} (\xi_2\bigr)\ind_{A^c}.
	\]
\end{itemize}

We note that for a time-consistent fully dynamic convex risk measure $\rho$, the residual risk after partial mitigation is also a time-consistent fully dynamic convex risk measure, see \cite[Section 4]{KS05}, and we write, for the buyer's and seller's perspective respectively,
\begin{align*}
	\hat{\rho}_{s,t}(\zeta) &:= \essinf_{C \in \mathcal{C}^\text{buy}_{s,t}} \rho_{s,t} (\zeta + C), \quad \zeta \in \loftp,\\
	\check{\rho}_{s,t}(\zeta) &:= \essinf_{C \in \mathcal{C}^\text{sell}_{s,t}} \rho_{s,t} (\zeta + C), \quad \zeta \in \loftp,
\end{align*}
for the residual risk from the buyers (resp. sellers) perspective. Specifically, we do not require that the risk measures are normalized and note that even if we assume normalization of $\rho$, i.e., $\rho_{s,t}(0)=0$ for all $0 \leq s\leq t \leq T$, this property in general does not carry over to $\hat{\rho}$ and $\check{\rho}$. Furthermore, we wish to point out that the essential infima (as all further essential infima and suprema in the text) are not required to be attained, i.e., there does not have to exist an optimizer $C^* \in \mathcal{C}^\text{buy/sell}_{s,t}$.

\subsection{Seller's Price}\label{sec:sell}

We start with the discussion of the seller's price. This is the more intricate problem as the exercise of the option is done by the \textit{buyer}, not the seller. Thus, the seller has to take into account any potential exercise action by the buyer. She cannot look at indifference at the time of exercise, as this time is not known to her -- it is a stopping time measurable with respect to the filtration $\mathbb{F}^\text{buy}$ but might not be measurable with respect to $\mathbb{F}^\text{sell}$. Additionally, the seller does not know the exercise time as a random variable, but only its realization along the realized path of assets. Moreover, the hedging strategy should, of course, be predictable in the appropriate filtration ($\mathbb{F}^\text{sell}$ enlarged by the realized exercise time) to make it practically implementable. The literature contains several definitions of the seller's price which, however, fall short of these requirements (we discuss details in Section \ref{sec:lit-comp}). We thus start setting up our definition from scratch:

Let $\mathcal{H}_\text{sell}$ be the set of predictable processes with respect to $\mathbb{F}^\text{sell}$, the information available to the seller. It represents the trading strategies that the seller can implement over the time interval $[0, T]$ without using the information about the actual exercise. From this set we now construct the set of all strategies that the trader can perform if they take the information of the exercise into account.

\begin{definition}
	Let $\mathcal I$ denote the set of possible exercise times of the option, which can either be $[0, T]$ or a closed countable subset of $[0, T]$. We define
	\begin{equation}\label{eq:Hprime}
		\mathcal H^\prime_\text{sell} := \bigl\{H: \mathcal I\times [0,T] \times \Omega \to \mathbb R:  \forall t\in \mathcal I, \, H(t, \cdot, \cdot)\in 							\mathcal{H}_{\text{sell}} \text{ and }   H(t_1, s , \cdot)=H(t_2, s, \cdot) \text{ for } s\leq t_1\wedge t_2\bigr\}.
	\end{equation}
	Additionally, all $H\in \mathcal{H}^\prime_{\text{sell}}$ must be right-continuous in the first variable, i.e.,
	\[
		\lim_{u \downarrow t}H(u, \cdot, \cdot)= H(t, \cdot, \cdot),
	\]
	when $t$ is a left limit point in $\mathcal I$.
\end{definition}

We refer to the set $\mathcal{H}^\prime_{\text{sell}}$ as the set of \textit{strategy selection functions}, as each $H\in \mathcal{H}^\prime_{\text{sell}}$ selects an initial strategy from $\mathcal{H}$ that gets updated at the point of exercise. Given any choice of strategy selection function, $H\in \mathcal H^\prime_\text{sell}$, we will denote by $h^\tau$ the strategy followed when $\tau \in \mathcal{T}_{0,T}$ is the exercise time, in other words, 
\[
	h^\tau_t(\omega):=H\bigl(\tau(\omega), t,\omega\bigr) \quad \text{ for } 0\leq t\leq T, \, \omega\in \Omega.
\]

Observe that, by definition, the strategies defined through selection functions in $\mathcal{H}^\prime_{\text{sell}}$ have the following ``non-anticipativity" property. For any $H\in \mathcal{H}^\prime_{\text{sell}}$ and $\tau_1, \tau_2\in \mathcal T_{0,T}$, 
\[
	H\bigl(\tau_1(\omega), s, \omega\bigr) = H\bigl(\tau_2(\omega), s, \omega\bigr), \text{ for }s\,\leq \, \tau_1(\omega)\wedge \tau_2(\omega).
\]
As a consequence, we get
\begin{equation}\label{eq:htau}
	h^\tau_t(\omega) = H\bigl(\tau(\omega), t, \omega\bigr)= H(T, t, \omega)\ind_{\{\tau(\omega)\geq t\}}+ H\bigl(\tau(\omega)\wedge t,t,\omega\bigr)\ind_{\{\tau(\omega)<t\}},
\end{equation}
for any $H\in \mathcal H^\prime_\text{sell}$. This decomposition implies that $h^\tau$ is a predictable process with respect to the filtration $\mathbb{F}^{\text{sell},\tau}$, $\mathcal{F}^{\text{sell}, \tau}_t = \bigcap_{\varepsilon >0}\mathcal{K}_{t+\varepsilon}$,  $\mathcal{K}_{t} = \mathcal{F}^\text{sell}_t \vee \sigma(\tau \wedge t)$, see \cite[Section 9]{Nik06}. We note that our construction follows in spirit the arguements of \cite[Section 2]{Kue02}; we adapted it to make it operational in the continuous choice setting

\begin{proposition}
	The process $h^\tau$ is $\mathbb{F}^{\text{sell},\tau}$-predictable.
\end{proposition}

\begin{proof}
	By definition, $H(T, \cdot, \cdot)$ is $\mathbb{F}^{\text{sell}}$-predictable, and using \eqref{eq:htau} we can write
	\[
		h^\tau_t(\omega) = H(T, t, \omega)+ \Bigl(H\bigl(\tau(\omega),t,\omega\bigr)- H(T, t, \omega)\Bigr)\ind_{\{\tau(\omega)<t\}},
	\]
	so it remains to show that 
	\[
		I_t:= \Bigl(H\bigl(\tau(\omega),t,\omega\bigr)- H(T, t, \omega)\Bigr)\ind_{\{\tau(\omega)<t\}}
	\]
	is $\mathbb{F}^{\text{sell}, \tau}$-predictable. We prove this by approximating $I_t$ by a sequence of $\mathbb{F}^{\text{sell}, \tau}$-predictable processes as 			follows.

	Define a sequence of stopping times $\tau_n$ by setting $\tau_n(\omega) = s_k^n := \sup\bigl\{ s \in \mathcal{I} \, : \, s< (k+1)2^{-n} \wedge T\bigr\}$ if $k2^{-n} 		\leq \tau(\omega) < (k+1)2^{-n}$, for $k=0, \hdots, \lfloor 2^n T\rfloor$. Then $\tau_n$ are $\mathbb{F}^\text{buy}$-stopping times taking finitely many values with 		$\tau_n \downarrow \tau$. Define moreover
	\[
		I^{(n)}_t:=  \sum_{k=0}^{\lfloor 2^nT\rfloor} \tilde{H}\bigl(\tau_n(\omega),t,\omega\bigr) \ind_{\{\tau_n(\omega)= s_k^n\}}
		\ind_{\{(k+1)2^{-n}<t\}},
	\]
	where $\tilde{H}\bigl(\tau_n(\omega),t,\omega\bigr) = H\bigl(\tau_n(\omega),t,\omega\bigr)- H(T, t, \omega)$ and observe that $I^{(n)}_t$ is a sum of
	$\mathbb{F}^{\text{sell}, \tau}$-predictable processes: By definition of $H\in \mathcal H^\prime_\text{sell}$, $\tilde{H}\bigl(s_k^n,t,\omega\bigr) $ is an
	$\mathbb{F}^{\text{sell}}$-predictable process (and hence $\mathbb{F}^{\text{sell},\tau}$-predictable) for each $k$. Moreover, for each $k$,
	\[
		\ind_{\{\tau_n(\omega)=s_k^n\}}\ind_{\{(k+1)2^{-n}<t\}} = \ind_{A^{(n)}_k}(t,\omega),
	\]
	where $A^{(n)}_k:= ((k+1)2^{-n}, T] \times \{\tau_n = s_k^n  \}\subset [0,T]\times \Omega$  is a set in the predictable $\sigma$-algebra $\mathcal{P}						(\mathbb{F}^{\text{sell},\tau})$ with respect to the filtration $\mathbb{F}^{\text{sell},\tau}$, as $\bigl\{\tau_n = s_k^n\bigr\} =\bigl\{k2^{-n} \leq \tau(\omega) 			\wedge (k+1)2^{-n}<(k+1)2^{-n}\bigr\}$ is $\sigma(\tau\wedge (k+1)2^{-n})$ measurable and hence $\mathbb{F}^{\text{sell},\tau}_{(k+1)2^{-n}}$ measurable.

	For each $(t, \omega)\in [0, T]\times \Omega$, 
	\begin{align*}
		&\phantom{==}I_t(\omega) - I^{(n)}_t(\omega)= \sum_{k=0}^{\lfloor 2^nT\rfloor} \Bigl[\tilde{H}\bigl(\tau(\omega),t,\omega\bigr)
		\ind_{\{\tau(\omega)<t\}}- \tilde{H}\bigl(\tau_n(\omega),t,\omega\bigr) \ind_{\{(k+1)2^{-n}<t\}}\Bigr]
		\ind_{\{k2^{-n} \leq \tau(\omega)< (k+1)2^{-n} \wedge T\} }\\
		& = \sum_{k=0}^{\lfloor 2^nT\rfloor} \Bigl[\tilde{H}\bigl(\tau(\omega),t,\omega\bigr)
		\bigl( \ind_{\{\tau(\omega)<t\leq (k+1)2^{-n}\}} + \ind_{\{(k+1)2^{-n}<t\}} \bigr) - \tilde{H}\bigl(\tau_n(\omega),t,\omega\bigr)
		\ind_{\{(k+1)2^{-n}<t\}}\Bigr] \ind_{\{k2^{-n} \leq \tau(\omega)< (k+1)2^{-n} \wedge T \} }\\
		& = \sum_{k=0}^{\lfloor 2^nT\rfloor} \Bigl(\tilde{H}\bigl(\tau(\omega),t,\omega\bigr)- \tilde{H}\bigl(\tau_n(\omega),t,\omega\bigr) \Bigr)
		\ind_{\{(k+1)2^{-n}<t\}}\ind_{\{k2^{-n} \leq \tau(\omega)< (k+1)2^{-n} \wedge T\} } \\
		&\qquad \qquad+ \sum_{k=0}^{\lfloor 2^nT\rfloor}  \tilde{H}\bigl(\tau(\omega),t,\omega\bigr)\ind_{\{k2^{-n} \leq \tau(\omega)< (k+1)2^{-n} \wedge T \} }  \ind_{\{\tau(\omega)<t\leq (k+1)2^{-n}\}} .
	\end{align*}
	Thus, $I_t(\omega)$ is the pointwise limit of $I^{(n)}_t(\omega)$ for each $(t, \omega)$ as $n\to\infty$, since  $\tau_n(\omega)\downarrow \tau(\omega)$ and by 			right-continuity of $s\mapsto H(s,t,\omega)$. 
\end{proof}

Note that while all process in $\mathcal{H}^\prime_\text{sell}$ are $\mathbb{F}^{\text{sell},\tau}$-predictable, not all $\mathbb{F}^{\text{sell},\tau}$-predictable processes are contained in $\mathcal{H}^\prime_\text{sell}$. Specifically, $\mathcal{H}^\prime_\text{sell}$ contains only the processes that are depending on the realization of the stopping time, not on the random variable itself.

Based on this clarification on the informational structure of the problem, we turn our attention to the seller's risk-indifference price for an American put. Here, the relevant set of hedging strategy selection functions in this case is $\mathcal{H}^\prime_{\text{sell}}$ as the seller has no prior knowledge of the exercise time of the put; she can only choose from hedging strategies that get updated after the exercise time is observed, thus $\mathcal{H}^\prime_{\text{sell}}$ is the set of strategy selection functions that is relevant.  For any choice of strategy $H\in \mathcal{H}^\prime_{\text{sell}}$, the seller will follow the strategy  $H(T, t, \omega)$ until $\tau(\omega)$, and after $\tau(\omega)$, she switches to the strategy $H\bigl(\tau(\omega),t,\omega\bigr)$ which is $\tau(\omega)$ dependent. We emphasize that this strategy $h_\cdot^\tau(\omega)$ is only $\tau(\omega)$ dependent and not $\tau$ dependent, i.e., depends only on the realization of the random variable in the specific scenario, not on $\tau$ as a random variable (cf. \cite[Section 2]{Kue02} for a discussion in the context of decision functions). This is a necessary distinction as the seller only observes $\tau(\omega)$ and may not glean any further knowledge of $\tau$. Given that strategy, the seller has to minimize the risk by considering the worst case over all stopping times. As she does not have insight into the buyer's information structure, she cannot optimize over $\mathcal{T}_{t,T}^\text{buy}$ but has to stick to the information available to her, i.e., use $\mathcal{T}_{t,T}^\text{sell}$.

Thus, for initial wealth $x$ and a time-consistent fully dynamic risk measure $\rho$, by definition of indifference pricing the seller's price $P^\text{sell}_t$ has to satisfy at time $t$
\begin{equation*}
	\check{\rho}_{t,T}(x)=\essinf_{H\in \mathcal H^\prime_\text{sell}}\esssup_{\tau\in \mathcal{T}_{t,T}^\text{sell}}\, \rho_{t,T}\biggl(x+\int_t^T h_s^\tau \,  dS_s -
	\xi_\tau+P_t^\text{sell}\biggr).
\end{equation*}
Solving for $P^\text{sell}$ and using cash invariance we get
\begin{equation}\label{S}\tag{S}
	P_t^\text{sell} = \essinf_{H\in \mathcal H^\prime_\text{sell}} \esssup_{\tau\in \mathcal{T}_{t,T}^\text{sell}}\, \rho_{t,T}\biggl(\int_t^T h_s^\tau \, dS_s -
	\xi_\tau\biggr) - \check{\rho}_{t,T}(0).
\end{equation}

\subsection{Buyer's Price}\label{sec:buy}

We are now turning our attention to the buyer's perspective. The buyer of the American claim has the right, but not the obligation to exercise the option at any time before maturity $T$. Therefore, she can decide on the $\mathbb{F}^\text{buy}$-stopping time $\tau$, which will provide her with a payout of $\xi_\tau$. Once the buyer with initial capital $x$ has bought the option at time $t$ for a price $P^\text{buy}_t$, she will of course try to reduce the risk of her position -- thus determine the exercise time $\tau^\text{buy}$ by minimizing the risk through both, determining the optimal exercise time and optimal hedging in the market until maturity of the option contract, i.e.,
\begin{equation*}
	\essinf_{\tau \in \mathcal{T}_{t,T}^\text{buy}} \hat{\rho}_{t,T} \bigl(x+ \xi_\tau - P^\text{buy}_t\bigr) =
	\essinf_{\tau \in \mathcal{T}_{t,T}^\text{buy}}\essinf_{C \in \mathcal{C}_{t,T}^\text{buy}} \rho_{t,T} \Bigl(x+ \xi_\tau - P^\text{buy}_t + C\Bigr).
\end{equation*}

Note that this minimization is independent of $P^\text{buy}_t$ and $x$ due to the cash-invariance property. We also point out that as the exercise time here is decided by the buyer, the knowledge about the exercise time does not add information, thus the hedging strategies are indeed $\mathbb{F}^\text{buy}$-predictable processes $\mathcal{H}_\text{buy}$. As a consequence of indifference pricing, and using cash invariance, we get the explicit representation 
\begin{align}\label{B}\tag{B}
	P_t^\text{buy} = \hat{\rho}_{t,T} (0) -  \essinf_{\tau \in \mathcal{T}_{t,T}^\text{buy}} \hat{\rho}_{t,T} \bigl( \xi_\tau\bigr).
\end{align}

\begin{remark}\label{rem:spread}
	If one considers the risk measures as being that of representative agents, one can understand $P_t^\text{buy}$ and $P_t^\text{sell}$ as bid and ask prices and their
	difference as bid-ask spread. In that case it is natural to assume that $\mathbb{F}^\text{buy} = \mathbb{F}^\text{sell}$ (e.g., both equal $\mathbb{F}^\text{S,B})$, so
	$\mathcal H^\prime_{\text{sell}} = \mathcal H^\prime_{\text{buy}} =: \mathcal H^\prime$ and $\check{\rho}$ and $\hat{\rho}$ agree. In this case, by choosing 
	$(\tau_n)$, $\tau_n \in \mathcal{T}_{t,T}^\text{buy}$, as a sequence of stopping times such that
	$ \hat{\rho}_{t,T} \bigl( \xi_{\tau_n}\bigr) \downarrow \essinf_{\tau \in \mathcal{T}_{t,T}^\text{buy}} \hat{\rho}_{t,T} \bigl( \xi_\tau\bigr)$, we get 
	\begin{align*}
		P_t^\text{sell} - P_t^\text{buy}  & = \Biggl(\essinf_{H\in \mathcal H^\prime}\esssup_{\tau\in \mathcal T_{t,T}}\, \rho_{t,T}\biggl(\int_t^T h_s^\tau \,  dS_s -
		\xi_{\tau}\biggr) - \hat{\rho}_{t,T} ( 0) \Biggr) - \Bigl( \hat{\rho}_{t,T} ( 0) - \essinf_{\tau \in \mathcal{T}_{t,T}^\text{buy}}
		\hat{\rho}_{t,T} \bigl( \xi_\tau\bigr)\Bigr)\\
		& = \lim_{n \to \infty}\Biggl(\essinf_{H\in \mathcal H^\prime}\esssup_{\tau\in \mathcal T_{t,T}}\, \rho_{t,T}\biggl(\int_t^T h_s^\tau \,  dS_s -\xi_{\tau}\biggr) - 
		\hat{\rho}_{t,T} ( 0) \Biggr) - \Bigl( \hat{\rho}_{t,T} ( 0) - \hat{\rho}_{t,T} \bigl( \xi_{\tau_n}\bigr)\Bigr)\\
		& \geq \lim_{n \to \infty} \Biggl(\essinf_{H\in \mathcal H^\prime}\, \rho_{t,T}\biggl(\int_t^T h_s^{\tau_n} \,  dS_s -\xi_{\tau_n}\biggr)  + \hat{\rho}_{t,T}
		\bigl( \xi_{\tau_n} \bigr) - 2\hat{\rho}_{t,T} ( 0) \\
		& = \lim_{n \to \infty}  \hat{\rho}_{t,T}\bigl(-\xi_{\tau_n}\bigr)  + \hat{\rho}_{t,T} \bigl( \xi_{\tau_n} \bigr) - 2\hat{\rho}_{t,T} ( 0) = \lim_{n \to \infty}  
		2 \Bigl( \frac{1}{2}\hat{\rho}_{t,T}\bigl(-\xi_{\tau_n}\bigr)  + \frac{1}{2}\hat{\rho}_{t,T} \bigl( \xi_{\tau_n} \bigr)\Bigr) - 2\hat{\rho}_{t,T} ( 0)\\
		& \geq \lim_{n \to \infty}  2 \hat{\rho}_{t,T}\Bigl(\frac{-\xi_{\tau_n}+ \xi_{\tau_n}}{2} \Bigr) - 2\hat{\rho}_{t,T} ( 0) =0, 
	\end{align*}
where the last inequality is by the convexity of the risk measure. This shows that the definition indeed yields a non-negative bid-ask spread.
\end{remark}

\subsection{Arbitrage}

We have to make sure that the notions of seller's and buyer's prices introduced above are free of arbitrage, i.e., a buyer/seller cannot make an arbitrage by buying/selling the option at the price $P^\text{buy/sell}$ and trading in the market. The next definition makes this notion precise.

\begin{definition}
	We define arbitrage from the buyer's and seller's perspective respectively.
	\begin{itemize}
	\item \label{def:sellersarbitrage}
		A price $p$ at time $t$ provides a \textit{seller's arbitrage opportunity} if there is a hedging strategy $\hat{H} \in \mathcal{H}^\prime_{\text{sell}}$ such that for 
		some amount $x < p$ and  for all stopping times $\tau \in \mathcal{T}_{t,T}^\text{sell}$ we have that
		\[
			x + \int \limits_{t}^{T} \hat{h}^\tau_{s} \,dS_{s} - \xi_{\tau} \geq 0,
		\]
		i.e., the seller can pocket the profit $p-x>0$ at time $t$ while being exposed to no risk of loss at time $T$.
	\item \label{def:buyersarbitrage}
		A price $p$ at time $t \in [0,T]$  provides a \textit{buyer's arbitrage opportunity} if there exists a hedging strategy $\hat{h} \in \mathcal{H}_{\text{buy}}$ together 
		with an exercise strategy $\tau \in \mathcal{T}^{\text{buy}}_{t,T}$ such that for some amount $x > p$, 
		\[
			-x + \int \limits_{t}^{T} \hat{h}_{s} \,dS_{s} + \xi_{\tau} \geq 0,
		\]
		i.e., the buyer can pocket the profit $x-p>0$ at time $t$ while being exposed to no risk of loss at time $T$.
	\end{itemize}
\end{definition}

\begin{proposition}\label{prop:arbitrage}
	The price defined by \eqref{S} is free of seller's arbitrage and the price defined by \eqref{B} is free of buyer's arbitrage.
\end{proposition}

\begin{proof}
	We first show that $p = P_{t}^{\text{sell}}$ does not allow for seller's arbitrage opportunities. Suppose by contradiction that there exists a seller's arbitrage opportunity
	$\hat{H}$ in the sense of Definition \ref{def:sellersarbitrage}.  Then  
	\begin{align*} 
		\essinf \limits_{H \in \mathcal{H}^\prime_\text{sell}}\,\esssup \limits_{\tau \in \mathcal{T}_{t,T}^\text{sell}} \,
		\rho_{t,T} \Biggl ( x  + \int \limits_{t}^{T}h_{s}^\tau\,dS_{s} - \xi_{\tau} \Biggr) 
		&= \essinf \limits_{H \in \mathcal{H}^\prime_\text{sell}}\,\esssup \limits_{\tau \in \mathcal{T}_{t,T}^\text{sell}} \, 
		\rho_{t,T} \Biggl(x  + \int \limits_{t}^{T}h_{s}^\tau\,dS_{s} + \int \limits_{t}^{T}\hat{h}^\tau_{s} \,dS_{s} - \xi_{\tau} \Biggr) \\
		&\leq \essinf \limits_{H \in \mathcal{H}^\prime_\text{sell}}\,\esssup \limits_{\tau \in \mathcal{T}_{t,T}^\text{sell}} \, 
		\rho_{t,T} \Biggl(  \int \limits_{t}^{T}h_{s}^\tau\,dS_{s}  \Biggr) \\
		&\leq \essinf \limits_{h \in \mathcal{H}_\text{sell}}\, \esssup \limits_{\tau \in \mathcal{T}_{t,T}^\text{sell}} \, \rho_{t,T} \Biggl(  \int \limits_{t}^{T}h_{s}\,dS_{s}  \Biggr) \intertext{ (by considering only constant strategy selection functions $H(u,s)=h_s$ for $h\in \mathcal H$)}
&= \essinf \limits_{h \in \mathcal{H}_\text{sell}}\, \rho_{t,T} \Biggl(  \int \limits_{t}^{T}h_{s}\,dS_{s}  \Biggr) = \check{\rho}_{t,T}(0).
	\end{align*} 
	whence
	\begin{equation*} 
		P_{t}^{\text{sell}} = \essinf \limits_{H \in \mathcal{H}^\prime_\text{sell}}\,\esssup \limits_{\tau \in \mathcal{T}_{t,T}^\text{sell}} \, 
		\rho_{t,T} \Biggl( \int \limits_{t}^{T}h_{s}^\tau\,dS_{s} - \xi_{\tau} \Biggr) - \check{\rho}_{t,T}(0) \leq x,
	\end{equation*}
	which contradicts the assumption that $x < P_{t}^{\text{sell}}$ and therefore disproves the existence of a seller's arbitrage opportunity.

	Similarly, assume a buyer's arbitrage opportunity $(\hat{h},\hat{\tau})$ exists in the sense of Definition \ref{def:buyersarbitrage}.  Employing a similar argument as above, 
	we find
	\[
		\hat{\rho}_{t,T}(x) \geq \hat{\rho}_{t,T}\Biggl(\int_t^T \hat{h}_s \, dS_s + \xi_{\hat{\tau}} \Biggr) = \hat{\rho}_{t,T}\bigl(\xi_{\hat{\tau}} \bigr)  \geq  
		\essinf \limits_{\tau \in \mathcal{T}_{t,T}^\text{buy}}  \hat{\rho}_{t,T}\bigl(\xi_\tau\bigr).
	\]
	Thus, cash invariance implies 
	\[
		P_{t}^{\text{buy}} = \hat{\rho}_{t,T}(0) -   \essinf \limits_{\tau \in \mathcal{T}_{t,T}^\text{buy}}  \hat{\rho}_{t,T}\bigl(\xi_\tau\bigr) \geq x.
	\] 
	We therefore conclude $x \leq P_{t}^{\text{buy}}$, contradicting the original assumption that $x > P_{t}^{\text{buy}}$ and thereby disproving the existence of a 
	buyer's arbitrage.
\end{proof}

Note that in the case the buyer's and seller's filtration agree, Remark \ref{rem:spread} also implies that the price \eqref{S} is free of buyer's arbitrage and \eqref{B} of seller's arbitrage. In the case the filtrations differ, no such argument can be made as the strategies live in different domains.

\subsection{Comparison with the Existing Literature}\label{sec:lit-comp}

The literature on indifference pricing of American options is long. Based on early results on European options in \cite{DPZ93}, Davis and Zariphopolou \cite{DZ95} explore utility indifference pricing in the presence of transaction costs, studying the singular control problem. This line of research has been deepened by Damgaard \cite{Dam06} and Zakamouline \cite{Zak05} who investigate the problem numerically for hyperbolic (resp. constant) absolute risk aversion, the latter adding the study of the seller's price. While all these papers assume asset prices given by geometric Brownian motion, Cosso, Marazzina and Sgarra \cite{CMS15} extend the buyer's side results to stochastic volatility. Oberman and Zariphopoulou, \cite {OZ03} use the indifference pricing methodology to price options on nontraded assets with dynamics correlated to traded assets from a buyer's perspective, using exponential utility in a geometric Brownian motion setting. An application to a regime switching model under expected utility indifference from the buyer's side of view is given by Gyulov and Koleva in \cite{GK22}.

Wu and Dai \cite{WD09} consider the indifference price of an American claim from a seller's point of view in a jump diffusion model under exponential utility. Bayraktar and Zhou \cite{BZ14} consider indifference pricing of American options on defaultable claims under exponential utility, for both buyer and seller. And K\"{u}hn \cite{Kue02} considers the problem of an option seller with a finite number of choices (such as Bermudan options) for general utility functions.

Two papers extend the problem to time dependent utilities. Leung, Sircar and Zariphopoulou \cite{LSZ12} consider forward performance measures and consider the buyer's indifference price in a stochastic volatility market, contrasting it to previous results for exponential utility in \cite{LS09}. Yan, Liang and Yang extend the indifference pricing setup in \cite{YLY15} to time dependent, additive stochastic differential utilities and optimal investment and consumption for an investor facing uncertainty about the risk-neutral probability measure. They discuss both seller's and buyer's perspectives.

We want to compare in particular the different notions of indifference price used. All papers, even those who consider both buyer's and seller's price, work with a single filtration setup. For the buyer's price, the papers \cite{DZ95, Dam06, Zak05, OZ03, LS09, LSZ12, GK22} use some form of (backward) stepwise maximization of strategies, after the exercise and before it, which is inspired from the Bellman principle in dynamic programming, solving the Merton problem from the exercise time onwards.

The definition of the buyer's and seller's price in \cite{BZ14} compares the expected utility of the hedged payoff for a given price at the time of the exercise with the utility of doing no investment at all. This notion strangely mixes notions of certainty equivalent and indifference price. But even if we adapt this notion in a way to compare buyer's risk at the time of exercise (determined to be risk-minimizing) and put necessary conditions that a minimal minimizing stopping time exists, this notion is in general not free of arbitrage. For the seller's price this approach is not even possible, as the potential exercise time is not known to the seller (only to the buyer).

For the seller's price the precise conditions on the admissibility of strategies are rarely fleshed out and most papers are cavalier about it. Zakamouline \cite{Zak05} uses an analogous version to the buyer's formulation, but assumes that the seller knows the optimal strategy of the buyer. K\"{u}hn \cite{Kue02} alone gives a careful discussion and a precise definition, albeit only for the discrete case with finitely many payoff options. Our definition is essentially a generalization of this framework to the general case. Note that a similar formulation of nonanticipativity was given in \cite{BHZ15} in the context of superhedging under model uncertainty. However, contrary to $\mathcal H^\prime_{\text{sell}}$ in \eqref{eq:Hprime} they consider not only the realization of the stopping times, but the stopping times (as random variables) themselves. As the option seller has only information on the actual exercise of the option, not hypothetical different asset price and exercise scenarios, we insist that the formulation should depend only on the realization of the stopping time, a point that K\"{u}hn rightfully highlights in \cite[Remark 2.3]{Kue02}. (A further slight difference is that we use $s\leq t_1\wedge t_2$ instead of $s< t_1\wedge t_2$ which assures $h^\tau$ to be a predictable process with respect to $\mathbb{F}^{\tau, \text{sell}}$.)

Finally, the forward performance formulation in \cite{LSZ12} and the stochastic differential utility in \cite{YLY15} are due to their recursive nature independent of the time horizon, avoiding in this way the intricacies of the dependence on the horizon of the indifference.

\section{Stochastic Volatility Models \& BDSE-R-BSDEs}\label{sec:stochvol}

We now turn our attention to a class of specific models to provide explicit representations of risk-indifference prices, following mainly the setup of \cite{SS15}. Specifically, we assume that the risk-free asset has a constant interest rate and thus $dB_t = r B_t \, dt$, $B_0 =1$. The price of the discounted risky asset is given by 
\begin{align*}
	dS_t & = \bigl(\mu(V_t) -r \bigr)S_t\, dt + \sigma(V_t) S_t\, dW_{1,t}, \quad S_0 = s,\\
	dV_t &= m(V_t) \, dt +a(V_t) \, dW_{2,t}, \quad V_0 = v,
\end{align*}
with correlated Brownian motions $W_1$, $W_2$ with constant correlation $\rho$. These models are very popular among practitioners and are usually called \textit{stochastic volatility models}. 

\begin{assumption}\label{model_assmpn}
We assume about the model parameters that
\begin{enumerate}
\item $\sigma, a\in C^{1+\beta}_{loc}(\R)$ for some $\beta>0$; 
\item there exist constants $\underline{\sigma}, \,\overline{\sigma}, \,\underline{a}, \,\overline{a}$ such that
\[0<\underline{\sigma}\leq \sigma \leq \overline{\sigma}<\infty, \qquad 0<\underline{a}\leq a \leq \overline{a}<\infty;\]
\item $\mu, m \in C^{0+\beta}_{loc}(\R)$ and $|\mu|<\overline{\mu}<\infty$.\\
\end{enumerate}
\end{assumption}
Additionally, we assume that
\begin{assumption}\label{terminal_cond}
The contingent claim $\xi_t:= f(S_t)$, where $f$ is a bounded continuous function. 
\end{assumption}
and 
\begin{assumption}\label{square_cond}
The admissible strategies satisfy the following square integrability condition
\[\E \biggl[\int_0^T h_s^2\sigma^2(V_s) S_s^2 ds\biggr] <\infty,\]
but for readability keep the notation of $\mathcal{H}_\text{buy/sell}$, now signifying the set of square integrable admissible strategies in the respective filtrations.
\end{assumption}

Moreover, we assume that both seller and buyer have no further information besides the asset prices, hence $\mathbb{F} = \mathbb{F}^{\text{buy}} = \mathbb{F}^{\text{sell}} = \mathbb{F}^{S,B}$. So, the (discounted) American claim $\xi$ is given by an almost surely continuous, bounded and $\mathbb{F}$-adapted process. We note that in principle this could be extended to non-bounded claims along the lines of \cite{HLT24} using the results in \cite{Sun21}.

We consider risk measures specified via solutions of backward stochastic differential equations (BSDEs). It is well-known (e.g., \cite{BEK09}) that if $g: \Omega \times [0,T] \times \mathbb{R}^{2} \to \mathbb{R}$ is a function satisfying certain properties (e.g., convexity on $\mathbb{R}^{2}$), then $g$ (called a driver) gives rise to a fully dynamic, strongly time-consistent monetary convex risk measure as the first component of the solution $(R,Z_{1}, Z_{2})$ to the BSDE given by
\begin{equation*}
	R_{t} = -\zeta - \int_t^u g(s, Z_{1,s},Z_{2,s})\,dt - \int_t^u Z_{1,s}\,dW_{1,s} -\int_t^u Z_{2,s}dW_{2,t}.
\end{equation*}
I.e., the time $t$ risk of a $\mathcal{F}_u$ measurable claim $\zeta$ at time horizon $u$ is given by $\rho_{t,u}(\zeta):= R_{t}$. For our purpose, we will assume that $g$ is just a function of $Z_{1}, Z_{2}$. What we are mostly interested in is not the risk itself, but the residual risk when we use hedging in the market to (partially) mitigate risk. For this purpose, we have to be a bit more restrictive.

\begin{definition} \label{def:driver}
	A driver $g : \mathbb{R}^{2} \to \mathbb{R}$ is called strictly quadratic with derivatives of (at most) linear growth if it satisfies
	\begin{enumerate}
		\item $g \in C^{2,1}(\mathbb{R}^{2})$;
		\item $g_{z_{1}z_{1}}(z_{1},z_{2}) > 0$ for all $z_{1},z_{2} \in \mathbb{R}$;
		\item there exists constants $c_{1}, c_{2} > 0$ such that
			\begin{equation*}
			c_{1}\Bigl(\frac{z_{1}^{2}}{4c_{1}^{2}} - (1 + z_{2}^{2}) \Bigr) \leq g(z_{1},z_{2}) \leq c_{2}(1 + z_{1}^{2} + z_{2}^{2});
			\end{equation*}
		\item there exists a constant $c_3>0$ such that $\frac{1}{c_{3}} \bigl( |z_{1}| - c_{3}(1 + |z_{2}|) \bigr) \leq |g_{z_{1}}(z_{1},z_{2})| \leq c_{3}(1 + |z_{1}| + |z_{2}|)$; 
		\item there exists a constant $c_4>0$ such that $|g_{z_{2}}(z_{1},z_{2})| \leq c_{4}(1 + |z_{1}| + |z_{2}|)$.
	\end{enumerate}
\end{definition}

Note that this notion is (slightly) more restrictive then the concept used in \cite{SS15}, relying only on conditions 1-3. In the American case we need the additional conditions as we rely, in the proof of the following theorem, on a comparison theorem for RBSDES (specifically \cite[Proposition 3.2]{KLQT02}) that requires them (cf. also \cite{Mor13} for a possible slight generalization). Hedging the risk is related to solving a BSDE with driver $g^*(-\lambda, z_2)+\lambda z_1$, where the principal part stems from a partial Fenchel-Legendre transform in the component that represents the tradeable instruments.

\begin{definition} 
	The nonlinear part $g^* : \mathbb{R}^{2} \to \mathbb{R}$ of the risk-adjusted driver is defined as the partial Fenchel conjugate of $g(z_{1},z_{2})$ in $z_{1}$, i.e., $g^*(\zeta,z_{2}) := \sup \limits_{z_{1} \in \mathbb{R}} \{ \zeta z_{1} - g(z_{1},z_{2})\}$.
\end{definition}

\begin{proposition}
	If $g$ is a strictly quadratic driver with derivatives of linear growth, then $g^*$ is also a strictly quadratic driver with derivatives of linear growth.
\end{proposition} 

\begin{proof}
	From Lemma 2.5 of \cite{SS15} we have that  the non-linear component $g^*$ of the risk-adjusted driver is strictly quadratic, if the driver $g$ is strictly quadratic. It remains to 			be shown that the non-linear component $g^*$ of the risk-adjusted driver satisfies conditions 4 and 5 of Definition \ref{def:driver} if $g$ does.
	
	Suppose condition 4 holds for $g(z_{1},z_{2})$.  Let $g_{z_{1}}^{-1}(z_{1},z_{2})$ denote the partial inverse of $g_{z_{1}}(z_{1},z_{2})$ in $z_{1}$.  Through the 
	simple variable change $z_{1} = g_{z_{1}}^{-1}(y,z_{2})$, the inequality $\frac{1}{c_{3}} \bigl( |z_{1}| - c_{3}(1 + |z_{2}|) \bigr) \leq |g_{z_{1}}|$ implies 
	$|g_{z_{1}}^{-1}| \leq c_{3}(1 + |z_{1}| + |z_{2}|)$, while the inequality $|g_{z_{1}}| \leq c_{3}(1 + |z_{1}| + |z_{2}|)$ implies 
	$\frac{1}{c_{3}}(|z_{1}| - c_{3}(1 + |z_{2}|)) \leq |g_{z_{1}}^{-1}|$, giving 
	\begin{equation} \label{eq:inverseinequality}
		\frac{1}{c_{3}}(|z_{1}| - c_{3}(1 + |z_{2}|)) \leq |g_{z_{1}}^{-1}| \leq c_{3}(1 + |z_{1}| + |z_{2}|).
	\end{equation}
	The desired inequality follows by noting $g_{z_{1}}^{-1} = g^*_{z_{1}}$.
	
	Next we suppose conditions 4 and 5 both hold for $g(z_{1},z_{2})$.  Employing the same variable change as in the argument involving condition 4 above, it follows from 
	condition 5 that $|g_{z_{2}}(g_{z_{1}}^{-1}(z_{1},z_{2}),z_{2})| \leq c_{4}(1 + |g_{z_{1}}^{-1}(z_{1},z_{2})| + |z_{2}|)$.  The result then follows from observing 
	that $g^*_{z_{2}}(z_{1},z_{2}) = -g_{z_{2}}(g_{z_{1}}^{-1}(z_{1},z_{2}),z_{2})$, and noting the upper bound obtained on $g^*_{z_{1}} = g_{z_{1}}^{-1}$ in 
	\eqref{eq:inverseinequality}.
\end{proof}

Before proving the main theorem, we want to recall a property of risk measures defined by Brownian BSDEs.

\begin{lemma}\label{lem:stc}
	In the BSDE setting, the strong time-consistency property holds for intermediate stopping times. Specifically, for all
	$\xi \in L^\infty(\Omega,\mathcal{F}_u, \mathbb{P})$, $s \leq u$, and $\tau \in [s,u]$ an $\mathbb{F}$-stopping time,
	\[
		\rho_{s,\tau} \bigl( -\rho_{\tau,u} (\xi)\bigr) = \rho_{s,u} (\xi).
	\]
\end{lemma}

\begin{proof}
	This follows from \cite[Theorem 3.21]{BEK09}.
\end{proof}

\begin{theorem}\label{thm:sellprice}
	For any $t\in[0,T]$, the seller's indifference price \eqref{S} can be represented as 
	\[
		P_t^{\text{sell}} = \check{R}^{\xi}_t - \check{R}^0_t,
	\]
	where $(\check{R}^\zeta, \check{Z}_1, \check{Z}_2, \check{K}, Y, \bar{Z}_1, \bar{Z}_2)$ is the unique solution to the BSDE-reflected BSDE (BSDE-R-BSDE) system
	\begin{align*}
		\left\{\begin{array}{ll}\check{R}^\zeta_u &= \zeta_T -\int_u^T \bigl(g^*(-\lambda_s, \check{Z}_{2,s}) + \lambda_s\check{Z}_{1,s}  \bigr) \, ds+
		\check{K}_T-\check{K}_u  -\int_u^T\check{Z}_{1,s} \, dW_{1,s} -\int_u^T\check{Z}_{2,s} \, dW_{1,s},\\
		Y_u &= 0 -\int_u^T \bigl( g^*(-\lambda_s, \bar{Z}_{2,s}) + \lambda_s\bar{Z}_{1,s}\bigr) \, ds  -\int_u^T\bar{Z}_{1,s} \, dW_{1,s} -\int_u^T\bar{Z}_{2,s} \, 
		dW_{2,s}\\
		\check{R}^\zeta_u &\geq  \zeta_u + Y_u \quad \text{with} \quad \int_t^T \bigl(\check{R}^\zeta_s - (\zeta_s+Y_s)\bigr) \, d\check{K}_s = 0, \qquad \text{ for }t\leq u\leq T,
		\end{array} \right.
	\end{align*}
	with $\zeta =\xi$ resp. $\zeta = 0$.
\end{theorem}

\begin{proof}
	Fix $t\in [0, T]$. The goal is to derive an RBSDE expression for
	\[
 		\essinf_{H\in \mathcal H^\prime} \esssup_{\tau\in \mathcal{T}_{t,T}}\, \rho_{t,T}\biggl(\int_t^T h_s^\tau \, dS_s -\zeta_\tau\biggr)
	\]
	for an almost surely continuous, bounded and $\mathbb{F}$-adapted process $\zeta$ for which we can then substitute $\zeta = \xi$ or $\zeta = 0$ to get the result. This is done in stages by proving several reformulations of the problem.

	We start by considering the claim $\zeta$ for a fixed hedging strategy selection function $H\in \mathcal H^\prime$ (suppressing the $\omega$-dependence of $H$ in the notation). Using 
	the strong time-consistency of Lemma \ref{lem:stc} and cash-invariance properties of risk measures, we can express the hedged risk of the American payoff $\zeta$ at 
	stopping time $\tau\in \mathcal T_{t,T}$ as
	\[
		\rho_{t,T}\Bigl(\int_t^Th_s^\tau \,dS_s -\zeta_\tau\Bigr)= \rho_{t,\tau}\biggl(\int_t^\tau H(T, s) \, dS_s-\zeta_\tau -\rho_{\tau, T}
		\Bigl(\int_\tau^T H(\tau, s) \, dS_s\Bigr )\biggr)\\
		= \rho_{t,\tau}\bigl(-U_{\tau}^{t,H}\bigr),
	\]
	where
	\[
		U_u^{t,H} := \zeta_u - \int_t^u H(T, s)\, dS_s +\rho_{u, T}\Bigl(\int_u^T H(u, s) \, dS_s\Bigr )
	\]
	for $u \in [t,T]$. Denote the supremum over all stopping times by 
	\[
		R^{t,H}_t:= \esssup_{\tau\in \mathcal T_{t,T}} \rho_{t,\tau}\bigl(-U_{\tau}^{t,H}\bigr).
	\]
	By \cite[Proposition 3.1]{KLQT02} we can represent $R^{t,H}_u = \esssup_{\tau\in \mathcal T_{u,T}} \rho_{t,\tau}\bigl(-U_{\tau}^{t,H}\bigr)$, for $t\leq u\leq T$, as 
	the first component of the (unique) solution of the RBSDE
	\[
		\left\{\begin{array}{ll}
		& R^{t,H}_u =  U_T^{t,H} +\int_u^T g\bigl(Z_{1,s}^{t,H}, Z_{2,s}^{t,H}\bigr)\, ds +K_T^{t,H} -K_u^{t,H}-\int_u^TZ^{t,H}_{1,s} \, dW_{1,s} - \int_u^T 
		Z_{2,s}^{t,H} \, dW_{2,s}, \\
		&R^{t,H}_u \geq U_u^{t,H}, \quad \int_t^T\bigl(R^{t,H}_r -U_r^{t,H}\bigr) \, dK_r^{t,H}=0, \quad t\leq u\leq T.
		\end{array} \right.
	\]
	Next, we have to consider $ \essinf_{H\in \mathcal H^\prime} R^{t,H}_t$. To do so, we first develop an alternative representation for the maximal risk. Define
	\[
		\ddot{R}^H_{u} := R^{t,H}_u + \int_t^u H(T, s) \, dS_s,
	\] 
	and note that it is the first component of the unique solution $\bigl(\ddot{R}^H_{u}, \ddot{Z}_{s}^{H,1}, \ddot{Z}_{s}^{H,2}, \ddot{K}_{u}^H\bigr)$ to the RBSDE
	\[
		\left\{\begin{array}{ll}
		\ddot{R}^H_{u} &=  \zeta_T - \int_u^T H(T, s) \, dS_s +\int_u^T g\bigl(\ddot{Z}_{1,s}^{H}, \ddot{Z}_{2,s}^{H}\bigr) \, ds +\ddot{K}_{T}^H -
		\ddot{K}_{u}^H-\int_u^T \ddot{Z}^{H}_{1,s} \, dW_{1,s} - \int_u^T \ddot{Z}_{2,s}^{H} \, dW_{2,s}, \\
		\ddot{R}^H_{u} &\geq \ddot{U}_{u}^H, \quad 
		\int_t^T\bigl(\ddot{R}^H_{r} -\ddot{U}_{r}^H\bigr) \, d\ddot{K}_{r}^H=0, \qquad \text{ for }t\leq u\leq T,
		\end{array} \right.
	\]
	where
	\[
		\ddot{U}_{u}^H := U_u^{t,H} + \int_t^u H(T,s) \, dS_s =  \zeta_u +\rho_{u, T}\Bigl(\int_u^T H(u, s) \, dS_s\Bigr)
	\]
	for $u \in [0,T]$. Observe that at time $u=t$,  $\ddot{R}^H_{t} = R^{t,H}_t$, and so we proceed to find a BSDE expression for
	$\essinf_{H\in \mathcal H^\prime} \ddot{R}^H_{t}$. 

	The reason for this alternative representation for the maximal risk, $\ddot{R}^H$, instead of $R^{t,H}$, now becomes clear:  for each $H\in \mathcal H^\prime$, the 
	BSDE dynamics for $\ddot{R}^H$ depends only on the strategy $H(T, \cdot)$ and its reflection barrier $\ddot{U}_u^H$ depends only on $H(u, \cdot)$, $u\in [t, T)$. We 
	can exploit this to separate the infima as follows:
 	\begin{equation}\label{separate-inf}
		\essinf_{H\in \mathcal H^\prime}\ddot R^H_u= \essinf_{H(T, \cdot)\in \mathcal H}\essinf_{\nu \in \mathcal H^{H(T,\cdot)}} \ddot{R}^\nu_{u},
	\end{equation}
	(while maintaining the right-continuity property of  $H\in \mathcal H^\prime$ in the first variable), where
	$\mathcal H^{H(T,\cdot)}:=\{\nu\in \mathcal H^\prime: \nu(T, \cdot)=H(T, \cdot)\}$.

	We first find an RBSDE representation for $\tilde{R}_u^H:=\essinf_{\nu \in \mathcal H^{H(T,\cdot)}}  \ddot{R}^\nu_{u}$. To this end, let us define for $u\in [t,T)$,
	\begin{equation}\label{inf_U}
		\bar{\zeta}_u := \essinf_{\nu \in \mathcal H^{H(T,\cdot)}} \ddot{U}_{u}^\nu = \essinf_{\nu(u, \cdot) \in \mathcal H} \ddot{U}_{u}^\nu  = \zeta_u +
		\check{\rho}_{u,T}(0)
	\end{equation}
	and let $Y_s :=\check{\rho}_{s,T}(0)$. Then from Proposition \ref{prop:optstrat} (cf. also \cite[Proposition 2.7]{SS15}), we have that $(Y, Z_1, Z_2)$ is the unique solution to the BSDE with terminal condition
	$\zeta_u$ and driver 
	\[
 		\inf_{\nu\in \mathbb{R}} \Bigl( - \nu\bigl(\mu(V_t)-r\bigr) + g\bigl(z_1-\nu\sigma(V_t), z_2\bigr)\Bigr)= -g^*\bigl(-\lambda_t, z_2\bigr) -z_1\lambda_t,
	\]
	with Sharpe ratio $\lambda_t = \frac{\mu(V_t)-r}{\sigma(V_t)}$. Thus,
	\[
		Y_u = 0 -\int_u^T \Bigl( g^*\bigl(-\lambda_s, \bar{Z}_{2,s}\bigr) +  \lambda_s\bar{Z}_{1,s}\Bigr) \, ds  -\int_u^T\bar{Z}_{1,s} \,dW_{1,s} -
		\int_u^T\bar{Z}_{2,s} \,dW_{2,s}.
	\]
	Following the arguments of Proposition \ref{prop:optstrat} (see \cite[Theorem 7.17]{KS05} for similar arguments in a superhedging setting) we see that the infimum in \eqref{inf_U}  is attained and the minimizing strategy is independent of $u$. By the
	comparison principle for quadratic RBSDEs (\cite[Proposition 3.2]{KLQT02}), we get that
	$\tilde{R}^H_{u} = \essinf_{\nu \in \mathcal H^{H(T,\cdot)}}  \ddot{R}^\nu_{u}$ satisfies
	\[
		\left\{\begin{array}{ll}
		\tilde{R}^H_{u} &=  \zeta_T - \int_u^T H(T, s) \, dS_s +\int_u^T g\bigl(\tilde{Z}_{1,s}^{H}, \tilde{Z}_{2,s}^{H}\bigr) \, ds +\tilde{K}_{t,T}^H -
		\tilde{K}_{u}^H-\int_u^T \tilde{Z}^{H}_{1,s} \, dW_{1,s} - \int_u^T \tilde{Z}_{2,s}^{H} \, dW_{2,s}, \\
		\tilde{R}^H_{u} &\geq \zeta_u + Y_u, \quad \int_t^T\bigl(\tilde{R}^H_{r} -(\zeta_r + Y_r)\bigr) \, d\tilde{K}_{r}^H=0.
		\end{array} \right.
	\]

	Finally, we take $\essinf_{H(T, \cdot) \in \mathcal H} \tilde{R}_{u}^H$, and, following the arguments of \cite[Theorem 7.17]{KS05}, get that
	\[
		\check{R}^\zeta_u : = \essinf_{H(T, \cdot) \in \mathcal H} \tilde{R}_{u}^H =\essinf_{H\in \mathcal H^\prime} \ddot{R}_t^H= \essinf_{H\in \mathcal H^\prime}
		R_t^{t,H}= \essinf_{H\in \mathcal H^\prime} \esssup_{\tau\in \mathcal{T}_{t,T}}\, \rho_{t,T}\biggl(\int_t^T h_s^\tau \, dS_s -\zeta_\tau\biggr)
	\] 
	has the representation given in the statement of the theorem. This concludes the proof.
\end{proof}

\begin{remark}\label{rem:bdry}
	We want to stress that the term $\zeta + Y$, $\zeta_u + Y_u = \zeta_u + \check{\rho}_{u,T}(0)$, appearing as reflection boundary has a clear economic interpretation: 
	One has to adapt the naive exercise boundary $\zeta$ by adding the (hedged) risk of the zero contract. Equivalently, as
	$\zeta_u + \check{\rho}_{u,T}(0) = \check{\rho}_{u,T}\bigl(-\zeta_u\bigr)$, one has to take the risk of the payment at the time of the exercise into account, however 
	allowing risk mitigation through trading until maturity.
\end{remark}

Analogously, but much easier, we can derive a RBSDE representation for the buyer's indifference price.

\begin{theorem}\label{thm:buyprice}
	The buyer's indifference price \eqref{B} can be represented as 
	\[
		P_t^{\text{buy}} = \hat{R}^\xi_t - \hat{R}^0_t
	\]
	where $(\hat{R}^\zeta, \hat{Z}_1, \hat{Z}_2, \hat{K}, Y, \bar{Z}_1, \bar{Z}_2)$ is the unique solution to the BSDE-R-BSDE
	\[
		\left\{\begin{array}{ll}
		\hat{R}^\zeta_u &=  \zeta_T +\int_u^T \bigl(g^*(-\lambda_s, -\hat{Z}_{2,s}) - \lambda_s\hat{Z}_{1,s}\bigr) \, ds +\hat{K}_T-\hat{K}_u -
		\int_u^T\hat{Z}_{1,s} \, dW_{1,s} -\int_u^T\hat{Z}_{2,s} \, dW_{2,s}, \\
		Y_u &=  0 +\int_u^T \bigl(g^*(-\lambda_s, -\bar{Z}_{2,s}) - \lambda_s\bar{Z}_{1,s}\bigr) \, ds  -\int_u^T\bar{Z}_{1,s} \, dW_{1,s} -
		\int_u^T\bar{Z}_{2,s} \, dW_{2,s}, \\
		\hat{R}^\zeta_u &\geq   \zeta_u + Y_u  \quad \text{with} \quad \int_t^T\bigl(\hat{R}^\zeta_r  -(\zeta_r + Y_r) \bigr) \, d\hat{K}_r^H=0.
		\end{array} \right.
	\]
	with $\zeta = \xi$ resp. $\zeta = 0$.
\end{theorem}

\begin{proof}
	Fix $t \in [0,T]$. We aim for an RBSDE expression for
	\[
		- \essinf_{\tau\in \mathcal{T}_{t,T}}  \essinf_{C \in \mathcal{C}_{t,T}} \, \rho_{t,T}\bigl(C +\zeta_\tau\bigr)
	\]
	for an almost surely continuous, bounded and $\mathbb{F}$-adapted process $\zeta$ for which we can then substitute $\zeta = \xi$ or $\zeta = 0$ to get the result. We 
	note first that $\tilde{Y}^{\zeta_\tau}$, 
	\[
		\tilde{Y}^{\zeta_\tau}_t := -\essinf_{C \in \mathcal{C}_{t,T}} \, \rho_{t,T}\bigl(C +\zeta_\tau\bigr) = -\hat{\rho}_{t,T}(\zeta_\tau) = 
		-\hat{\rho}_{t,\tau}\bigl( - \hat{\rho}_{\tau,T}(\zeta_\tau)\bigr),
	\]
	(by Lemma \ref{lem:stc}) satisfies the BSDE
	\[
		\tilde{Y}^{\zeta_\tau}_t =   - \hat{\rho}_{\tau,T}(\zeta_\tau)   +\int_t^\tau \Bigl( g^*\bigl(-\lambda_s, -\tilde{Z}_{2,s}\bigr) -\lambda_s\tilde{Z}_{1,s}\Bigr) \,
		ds  -\int_t^\tau\tilde{Z}_{1,s} \, dW_{1,s} -\int_t^\tau\tilde{Z}_{2,s} \, dW_{2,s}
	\]
	following Proposition \ref{prop:optstrat} (cf. also \cite[Proposition 2.7]{SS15}). Now (\cite[Proposition 3.1]{KLQT02}) implies that $\hat{R}^\zeta$, 
	\[
		\hat{R}^\zeta_t := -\essinf_{\tau\in \mathcal{T}_{t,T}} -\tilde{Y}^{\zeta_\tau}_t  = \esssup_{\tau\in \mathcal{T}_{t,T}} \tilde{Y}^{\zeta_\tau}_t,
	\]
	has with $\bar{\zeta}_u := - \hat{\rho}_{u,T}(\zeta_u)$ the RBSDE representation
	\[
		\left\{\begin{array}{ll}
		\hat{R}^\zeta_u &=  \bar{\zeta}_T +\int_u^T \bigl(g^*(-\lambda_s, -\hat{Z}_{2,s}) - \lambda_s\hat{Z}_{1,s}\bigr) \, ds +\hat{K}_T-\hat{K}_u -
		\int_u^T\hat{Z}_{1,s}\,  dW_{1,s} -\int_u^T\hat{Z}_{2,s} \, dW_{2,s}, \\
		\hat{R}^\zeta_u &\geq  \bar{\zeta}_u  , \quad 
		\int_u^T\bigl(\hat{R}^\zeta_s - \bar{\zeta}_s \bigr) \, d\hat{K}_s^H=0, \quad \bar{\zeta}_u := -\hat{\rho}_{u,T}(\zeta_u).
		\end{array} \right.
	\]
	Noting now that $ \bar{\zeta}_u = \zeta_u-\hat{\rho}_{u,T}(0)$ and writing down the BSDE representation of $Y := \tilde{Y}^{0}$ using \cite[Proposition 2.7]{SS15} 
	gives the result.
\end{proof}

\begin{remark}\label{rem:boundary}
	Note that in the case $\zeta = 0$, in both Theorem \ref{thm:sellprice} and \ref{thm:buyprice} the RBSDE and boundary dynamics agree, thus reducing the equation to a 
	classical (non-reflected) BSDE.
\end{remark}

\begin{remark}\label{rem:BS}
	We want to point out that the special case $\mu(v) \equiv \mu$, $\sigma(v) \equiv \sigma$ leads to the classical Black--Scholes model, a \textit{complete} market model. In this case both the buyer's and the seller's price agree with the arbitrage-free price. This is easy to see: in this case, the nonlinear component of the risk-adjusted drivers, $g^*$, does not depend on the second argument, and is deterministic. The (R)BSDEs reduce to $\hat{R}_t^\zeta:= R_t^\zeta +\int_t^Tg^*(-\lambda)ds$, where  $R^\zeta_t$ is an RBSDE with linear driver $-\lambda z$ and boundary $\zeta$, thus $P^{buy}_t= R_t^\xi -R_t^0=R_t^\xi$. Similarly, $\check{R}_t^\zeta:= R_t^\zeta -\int_t^Tg^*(-\lambda)ds$, giving us $P^{sell}_t= R_t^\xi = P^{buy}_t$.  Finally, as this price is free from buyer's and seller's arbitrage by Proposition \ref{prop:arbitrage}, it is the unique arbitrage free price. 
\end{remark}

\begin{remark}\label{rem:specifiv}
In the current section we have only considered the case where the buyer's and seller's filtration are equal, i.e., $\mathbb{F} = \mathbb{F}^{\text{buy}} = \mathbb{F}				^{\text{sell}} = \mathbb{F}^{S,B}$. This is due to the fact that on the one hand the main feature we want to highlight is that the asymmetry between buyer and seller already comes from the timing of the exercise, which is buyer determined. Moreover, there is no general way of including information asymmetry and the result will always be model dependent. Below, we discuss how the work can be extended to models with asymmetric filtrations by highlighting the features that would change. 

We first note that, in general, the difference in information will inform a different choice of risk measures, as they have to take the respective information into account. For example, if we include additional assets only available to one of the traders (or just information on these assets) via additional stochastic differential equations with correlated Brownian motions, the risk measures have to include that information.

To give a specific example, let's assume there exists an additional tradeable asset $U$ to which only the buyer has accesses, 
\[
	dU_t = \bigl(\nu(U_t) -r\bigr) U_t \, dt + \varsigma(U_t) U_t \, dW_{3,t}, \qquad U_0 =u
\] 
with a Brownian motion $W_{3,t}$ correlated to $W_{1,t}$ and $W_{2,t}$. In that case the buyer would have to use a risk measure with a driver depending on three variables, and the risk of holding the position $\zeta$ would be given by
\begin{equation*}
	R^B_{t} = -\zeta - \int_t^u g^B(s, Z^B_{1,s},Z^B_{2,s}, Z^B_{3,s})\,dt - \int_t^u Z^B_{1,s}\,dW_{1,s} -\int_t^u Z^B_{2,s}dW_{2,t} -\int_t^u Z^B_{3,s}dW_{3,t},
\end{equation*}
while the seller would use the driver with two variables,
\begin{equation*}
	R^S_{t} = -\zeta - \int_t^u g^S(s, Z^S_{1,s},Z^S_{2,s})\,dt - \int_t^u Z^S_{1,s}\,dW_{1,s} -\int_t^u Z^S_{2,s}dW_{2,t}.
\end{equation*}
Then, when it comes to hedging, one would have to specify if the buyer can actually trade in $U$ or only has information about $U$, to set up the relevant optimization over trading strategies.
 
To give a specific example of an application we have in mind: A subset of authors of the current paper are working on an example of how this approach can be implemented for the specific problem of shadow trading: insiders that want to stay within the bounds of legality can use their superior information to trade in derivatives written on stock of companies that are correlated to the insider's without violating the ban in trading in shares of their own company. Here the insider information on a non-tradeable asset enlarges the buyer's information, while the seller does not have this information.
\end{remark}

\section{Numerical Solution via Deep Learning}\label{sec:deep}

Solving the BSDE-R-BSDE systems of Theorems \ref{thm:sellprice} and \ref{thm:buyprice} is not a straightforward task, as we are encountering a four-dimensional problem with two forward and two backward SDEs, one of the backward ones serving as reflection boundary for the other. We rely on the recent breakthroughs in deep learning methods to solve this problem numerically. We first describe our general approach and then provide explicit solutions to a sample problem. The code for this implementation can be found at \url{https://github.com/stesturm/American-risk-indifference}.

\subsection{Solving BSDE-R-BSDEs using the RDBDP Method}

To solve the BSDE-R-BSDE systems of Theorems \ref{thm:sellprice} and \ref{thm:buyprice}, we rely on the Reflected Deep Backward Dynamic Programming (RDBDP) algorithm developed by Hur\'{e}, Pham and Warin in \cite{HPW20} (provided with more context and discussed in Hur\'{e}'s PhD thesis \cite{Hur19} as well as by Kharroubi in \cite{Kha21}). We show here the implementation for the seller's price (see Theorem \ref{thm:sellprice}), the one for the buyer's works analogously.

We divide the interval $[0,T]$ equidistantly by a partition $\pi$ of $N$ intervals, setting $t_i = \frac{iT}{N}$ for $i \in \{0, 1, \ldots N\}$. For the seller of an American style claim $h(S_t)$ we solve, backward iteratively, the system
\begin{align*}
	&\min_{\phi_{0,i}, \phi_{1,i}, \phi_{2,i} \in \mathcal{N}_i} \E\biggl[\Bigl\vert \phi_{0,i}(S^\pi_{t_i}, V^\pi_{t_i}) - \Bigl(Y^\pi_{t_{i+1}} - \bigl(g^*(-
	\lambda(V^\pi_{t_i}),  \phi_{2,i}(S^\pi_{t_i}, V^\pi_{t_i})) + \lambda(V^\pi_{t_i})\phi_{1,i}(S^\pi_{t_i}, V^\pi_{t_i})\bigr)  \, \Delta s \\ 
	&\phantom{==========}- \phi_{1,i}(S^\pi_{t_i}, V^\pi_{t_i}) \, \Delta W_{1,i} - \phi_{2,i}(S^\pi_{t_i}, V^\pi_{t_i}) \, \Delta W_{2,i}\Bigr)\Bigr\vert^2\biggr],\\
	&\min_{\hat{\phi}_{0,i}, \hat{\phi}_{1,i}, \hat{\phi}_{2,i} \in \hat{\mathcal{N}}_i} \E\biggl[\Bigl\vert \hat{\phi}_{0,i}(S^\pi_{t_i}, V^\pi_{t_i}) -
	\max\Bigl(\zeta_{t_i} + Y^\pi_{t_i}, \hat{R}^\pi_{t_{i+1}} - \bigl(g^*(-\lambda(V^\pi_{t_i}), \hat{\phi}_{2,i}(S^\pi_{t_i}, V^\pi_{t_i})) +
	\lambda(V^\pi_{t_i})\hat{\phi}_{1,i}(S^\pi_{t_i}, V^\pi_{t_i})\bigr)  \, \Delta s \\
	&\phantom{==========}- \hat{\phi}_{1,i}(S^\pi_{t_i}, V^\pi_{t_i}) \, \Delta W_{1,i} - \hat{\phi}_{2,i}(S^\pi_{t_i}, V^\pi_{t_i}) \, \Delta W_{2,i}\Bigr)
	\Bigr\vert^2\biggr]\\& \quad \text{subject to}\\
	&\left\{\begin{array}{ll}
		S^\pi_{t_{i+1}}  &= \bigl(\mu(V^\pi_{t_i}) -r \bigr)S^\pi_{t_i} \Delta t+\sigma(V^\pi_{t_i}) S^\pi_{t_i}\, \Delta W_{1,i}, \quad S^\pi_0 = s,\\
		V^\pi_{t_{i+1}} &= b(V^\pi_{t_i}) \, \Delta t +a(V^\pi_{t_i}) \, \Delta W_{2,i}, \quad V^\pi_0 = v,\\
		Y^\pi_{t_i} & = \phi^*_{0,i}(S^\pi_{t_i}, V^\pi_{t_i}), \quad Y^\pi_T =  0,\\
		\hat{R}^\pi_{t_i} & =  \hat{\phi}^*_{0,i}(S^\pi_{t_i}, V^\pi_{t_i}),  \quad \hat{R}^\pi_T =  h(S^\pi_T).
	\end{array} \right.
\end{align*}
where $\Delta W_{1,i}$, $\Delta W_{2,i}$ are the Brownian increments from time $t_i$ to $t_{i+1}$ and $\mathcal{N}$, $\hat{\mathcal{N}}$ are the hypothesis spaces for the deep neural networks for the boundary condition resp. the RBSDE (with one dimension for the solution process and two for the adjoint processes each) and $\phi^*_{0, i}$, $\hat{\phi}^*_{0, i}$ the stepwise optimizers (at time step $i$) of the first component. Practically, we first calculate the boundary condition for all time steps by calculating the solution of the zero terminal condition BSDE process and then adding it to the payoff at early exercise. Using the same Brownian paths by fixing seeds, we calculate the RBSDE process using the boundary condition previously calculated.  In this way, we have to solve only a single BSDE for the boundary, that we can use for the seller's  price of any type of American payoff  (cf. Remark \ref{rem:boundary}). Specifically, we use a deep neural network with $2$ hidden layers and ReLu activation functions, and use the Adam optimizer (cf. \cite{KB17}) for stochastic gradient descent. The implementation \url{https://github.com/stesturm/American-risk-indifference} is in TensorFlow.

\subsection{Numerical Illustration}

To illustrate the the results, we will consider a specific example along the lines of \cite[Section 3]{SS15} which allows for the direct comparison to the European option example considered there. We assume a classical American put option claim $\hat{\xi}$, thus the discounted claim is $\xi_t = \bigl(e^{-rt}K- S_t\bigr)^+$, and use distorted entropic risk measures, given by the driver
\[
	g\bigl(z_1, z_2\bigr) := \frac{\gamma}{2}\Bigl(z_1^2 + z_2^2\Bigr) + \eta \gamma z_1 z_2 + \frac{\eta^2 \gamma}{2}z_2^2
	= \frac{\gamma}{2}\Bigl(\bigl(z_1+\eta z_2\bigr)^2 + z_2^2\Bigr).
\]
This driver represents in the case $\eta=0$ a classic entropic risk measure (equivalent to exponential utility) with risk tolerance parameter $\gamma$; the term $\eta$ introduces an additional volatility risk premium. The Fenchel-Legendre transform is given by
\[
	g^*\bigl(z_1, z_2\bigr) := \frac{1}{2\gamma}\Bigl(z_1^2 -  \eta z_1 z_2 - \frac{\gamma}{2}z_2^2\Bigr).
\]
As stochastic volatility model we choose the arctangent model
\begin{align*}
	dS_t & = \bigl(\mu -r\bigr) S_t\, dt + \sigma(V_t) S_t\, dW_{1,t}, \quad S_0 = s,\\
	\sigma(y) & = \frac{a}{\pi}\Bigl(\arctan(y - 1) + \frac{\pi}{2}\Bigr) + b,\\ 
	dV_t &= \alpha \bigl( m-V_t\bigr) \, dt + \nu \sqrt{2\alpha} \Bigl(\rho \, dW_{1,t} +  \sqrt{1-\rho^2}\, dW_{2,t}\Bigr), \quad Y_0 = y.
\end{align*}
To display the results, we use the market convention to plot in addition to prices (European) implied volatilities by inverting the Black--Scholes formula. We choose as parameters
\begin{align*}
	r &= 0.02, \, \mu = 0.08,  \, a = 0.7, \, b = 0.03, \, s= 100, \, m=0, \, \alpha = 5, \, \nu = 1, \, \rho = -0.2, \, y = .15, \\  \gamma &= 1, \,\eta = 0.2, \, T = 0.25.
\end{align*}
For the hyperparameters of the neural network, we use adaptive epochs, namely 1000 in the beginning and 300 for the last 5 steps, at a batch size of 1100 and a learning rate of 0.01, and we use $N=10$ time steps. We are calculating the prices and the implied volatility for strikes from $K=85$ to $K=115$ in steps of $5$ and plot them against strikes resp. log-moneyness, see Figure \ref{fig:smiles}. We note as comparison that the initial volatility of the stock is $y \approx 15\%$ while the mean-reversion level is $\sigma(m) \approx 20.50\%$. 

\begin{figure}[htb]
	\centering
	\includegraphics[width=0.3\textwidth]{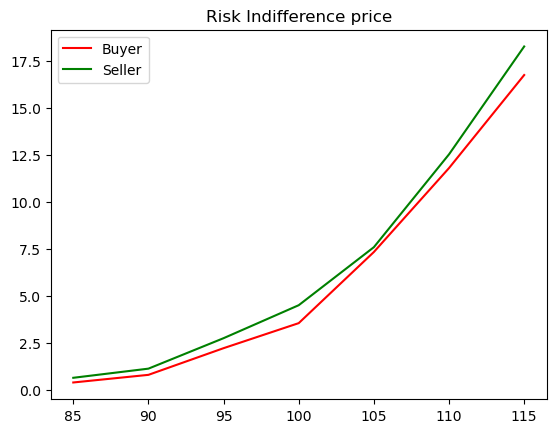}\quad 
	\includegraphics[width=0.3\textwidth]{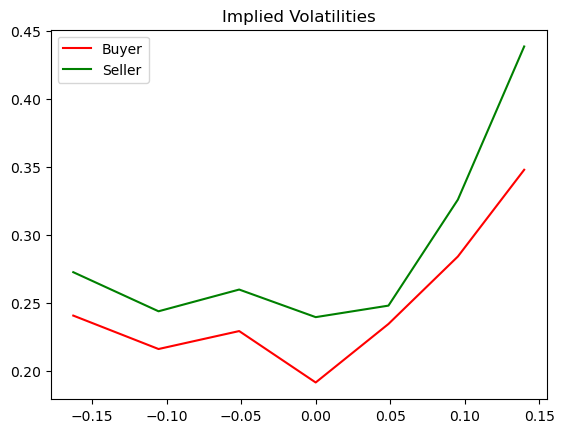}\quad
	\includegraphics[width=0.3\textwidth]{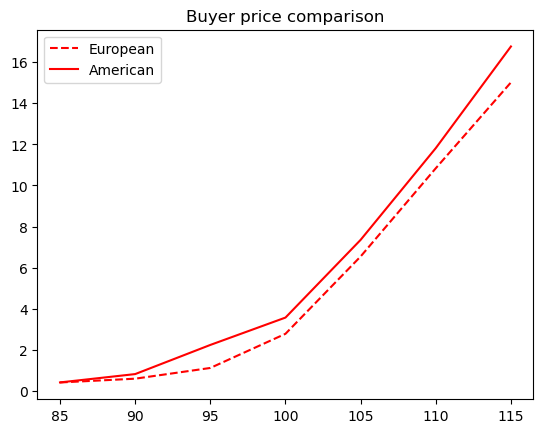}
	\caption{Left panel: Buyer's and seller's prices of American options in terms of strikes. Middle panel: Implied volatility of buyer's and seller's prices of American options in terms of log-moneyness. Right panel: Prices of buyer's prices of European and American options in terms of strikes.}\label{fig:smiles}
\end{figure}

Figure \ref{fig:smiles} shows in the left panel that the difference between buyer's and seller's indifference price is quite small as absolute (dollar) amount. This indicates that the hypothesis that seller's and buyer's price correspond to market ask and bid when the risk measure is consider as that of a representative agent is very reasonable. The middle panel shows that the differences between buyer and seller become more pronounced when we consider implied volatilities. Both, buyer an seller implied volatility curves exhibit the typical smile pattern observed in markets, and the curves allow for comparison with the results in \cite[Figure 3.4]{SS15} derived with traditional numerical methods. Finally, the right panel shows that the difference between European and American indifference prices is quite small (but larger than the difference between American buy and sell prices).

\section{Conclusion}\label{sec:conc}

Indifference pricing is an important mechanism to establish reasonable reservation prices for buyers and sellers of derivative claims. The current paper explains how this can be done for American-style claims using residual risk after hedging as indifference price mechanism, a choice that is driven by both, the availability of a comprehensive mathematical framework (risk measures and BSDEs) as well as the prevalence of risk measures in industrial practice (as compared to utility-based concepts).

The main contribution of the paper is twofold: On the one hand we provide a general and detailed setup for risk-indifference pricing of American style contingent claims; on the other hand we show how in the case of market incompleteness due to stochastic volatility, the risk-indifference price can be expressed through BSDE-R-BSDEs, backward stochastic differential equations in which the reflection boundary is given itself by a backward stochastic differential equation, reflecting the risk of the position between exercise and maturity. As an add-on, we show how the arising BSDE-R-BSDEs can be solved numerically using deep learning methods and illustrate this on a specific example.

\bibliographystyle{plain}
\bibliography{bibliography}
\appendix
\section{Appendix}

In this appendix we provide the necaessary result of the existence of the optimal strategy used in the proofs of Chapter \ref{sec:stochvol}, focusing first on BSDEs and then RBSDEs.

\subsection{The BSDE case}

Let $\bigl(R^*, Z_1^*, Z_2^*\bigr)$ denote the solution to the BSDE with terminal condition $\xi_T=f(S_T)$ and driver $- g^*(-\lambda_s, z_2)-\lambda_s z_1 = \inf_{\theta\in \R}\bigl\{-\theta\lambda_s+g( z_1-\theta, z_2)\bigr\}$, such that $R^*$ is bounded and $\E\bigl[\int_0^T \bigl(|Z_1^*|^2 +|Z_2^*|^2\bigr)\bigr]<\infty$.  Existence and uniqueness of this BSDE is provided by \cite{Kob00}.

\begin{proposition}\label{prop:optstrat}
	Let $f$ be a bounded continuous function. The minimal hedged risk associated with $\xi_T=f(S_T)$ is given by $R^*$. In other words, 
	\[
		R^*_t= \inf_{h\in \mathcal H} \rho_{t, T}\bigl(\xi_T + \int_t^T h_s \, dS_s\bigr).
	\]
	Moreover, the optimal strategy, $h^*$, is given by
	\begin{equation}\label{inf-strategy}
		h^*_s:= \frac{1}{\sigma(V_s)S_s}\Bigl(Z^*_{1,s} - g_{z_1}^{-1}\bigl(-\lambda_s, Z^*_{2,s}\bigr)\Bigr),
	\end{equation}
	i.e., $R_t^*= \rho_{t, T}\bigl(\xi_T + \int_t^T h^*_s \, dS_s\bigr)$. 
\end{proposition}

Note that we do not have necessarily $h^* \in \mathcal H$.

\begin{proof}
We will first identify the strategy $h^*$ such that $\rho_{t, T}\bigl(\xi_T + \int_t^T h^*_sdS_s\bigr)= R_t^*$.
For $h\in \mathcal H$, denote by $(R^h, Z_1, Z_2)$ the solution of the BSDE
\[
	R^h_t = - \xi_T -  \int_t^T h_s \, dS_s+\int_t^Tg(s, Z_{1,s}, Z_{2,s}) \, ds -\int_t^T Z_{1,s} \, dW_{1,s}  -\int_t^T Z_{2,s} \, dW_{2,s}.
\]
From \cite{Kob00} it follows that a solution exists, is unique and $R^h$ is bounded. Since the risk measures $\rho_{t,T}(\cdot)$ are given by BSDEs with driver $g$, we have $R_t^h= \rho_{t, T}\bigl(\xi_T + \int_t^T h_sdS_s\bigr)$. Rewriting this, we get 
\begin{align*}
	R^h_t &=-\xi_T+ \int_t^T \Bigl(-h_s\bigl(\mu(V_s)-r\bigr)S_s+g\bigl(s, Z_{1,s}, Z_{2,s}\bigr)\Bigr) \, ds -\int_t^T \Bigl(h_s \sigma(V_s)S_s +Z_{1,s}\Bigr) \, dW_{1,s}  -\int_t^T Z_{2,s} \, dW_{2,s} 
	\intertext{and setting $\theta_s:=h_s \sigma(V_s)S_s$ and   $\tilde{Z}_{1,s}:=\theta_s +Z_{1,s}$ finally}
	&= - \xi_T +\int_t^T \Bigl(-\theta_s \lambda_s+g\bigl(s, \tilde{Z}_{1,s}-\theta_s, Z_{2,s}\bigr)\Bigr) \, ds -\int_t^T \tilde{Z}_{1,s} \, dW_{1,s}  -\int_t^T Z_{2,s} \, dW_{2,s}.
\end{align*}
Observe that, by the comparison principle for BSDEs, see \cite[Theorem 2.6]{Kob00}, every $R^h_t$ is bounded below by $R_t^*$, the BSDE with driver  $\inf_{\theta\in \R}\bigl\{-\theta \lambda+g(s, z_1-\theta, z_2)\bigr\}$. The infimum is attained at $\theta^*_s:= z_1 - g_{z_1}^{-1}(-\lambda_s, z_2)$, and so $R_t^*= R_t^{h^*}$. However, this $h^*$ need not be in $\mathcal H$. So we only have $\inf_{h\in \mathcal H} R_t\geq R_t^*$ and not equality. 
By Lemma \ref{lem:limit_of_admissible_risk},  $R_t^* = \lim_{m\to\infty}R_t^{h^m}$ a.s., for a sequence of $h^m\in \mathcal H$  thus completing the proof.
\end{proof}

\begin{lemma}\label{lem:limit_of_admissible_risk}
There exists a sequence $\{h^m\}\subset \mathcal H$ such that $R_t^* = \lim_{m\to\infty}R_t^{h^m}$ a.s..
\end{lemma}

\begin{proof}
Under Assumption \ref{model_assmpn}, Theorem V.8.1 and Remark V.8.1 in \cite{Lad68} ensure the existence of a bounded continuous function $u(t,x,v)$ with bounded and continuous  first and second order derivatives  such that $R_t^*= u(t, X_t, V_t)$, $Z^*_{1, t}=\sigma(V_t)\p_x u(t, X_t, V_t)+\rho a(V_t)\p_v u(t,X_t, V_t) $ and $Z^*_{2,t}=\sqrt{1-\rho^2}a(V_t)\p_vu(t, X_t, V_t)$ (see Theorem 2.9 and its proof in \cite{SS15}), where $X_t:=\ln S_t$.  Thus, there exists a constant $C^*>0$ such that 
\begin{equation}\label{bound_Z^*}
	\bigl\vert R_t^*\bigr\vert, \bigl\vert Z^*_{1,t}\bigr\vert, \bigl\vert Z^*_{2,t}\bigr\vert <C^*, \qquad \text{for all }t\in [0,T].
\end{equation}
Then by the definition of $h^*$ in \eqref{inf-strategy}, we have $h^*S$ is also bounded. Without loss of generality let's say
\begin{equation}\label{eq:h*bound}
	\sup_{0\leq s\leq T}\bigl\vert h^*_sS_s \bigr\vert < C^*.
\end{equation}
By assumption on the coefficients for the stock price model, we have $\E\bigl[\int_0^T S_t^4 dt\bigr]<\infty$ and  $\E\bigl[\bigl(\int_t^uh_s^*d \, S_s\bigr)^2\bigr]<\infty$ for all $u\in [t,T]$. Consequently,
\begin{equation}\label{eq:finite-hedge}
	\biggl\vert \int_t^uh_s^* \, dS_s \biggr\vert <\infty \qquad  a.s. \text{ for all }u\in [t,T].
\end{equation}

Let now
\[
	\sigma_m:= \inf_{s\geq t}\biggl\{\biggl\vert \int_t^s h^*_u \, dS_{u}\biggr\vert >m\biggr\}\wedge T
\]
and define
\[
	h^m_s:= h^*_s \ind_{\{s\leq \sigma_m\}}.
\]
Then ,by \eqref{eq:finite-hedge}, we have  $\sigma_m\to T$ a.s. Also, $\bigl\vert \int_t^Th^m_s \sigma_s \, dS_s \bigr\vert \leq m$ and so $h^m$ is an admissible strategy and $h^m\to h^*$ a.s. Let $R^m_t:=R^{h^m}_t= \rho_{t,T}\bigl(X+\int_t^T h^m_s \, dS_s\bigr)$, then $\bigl(R^m, Z^m_1, Z^m_2\bigr)$ is a solution to the BSDE
 \[
	 R_t^m = -\xi_T +\int_t^T\Bigl(-h^m_s\sigma(V_s)S_s\lambda_s+ g\bigl(Z_{1,s}^m-h^m_s\sigma(V_s)S_s, Z_{2,s}^m\bigr)\Bigr) \, ds -\int_t^T Z^m_{1,s} \, dW_{1,s} -\int_t^T Z^m_{2,s}\, dW_{2,s}
 \]
 and recall that $\bigl(R^*, Z_1^*, Z_2^*\bigr)$ satisfies
 \[
 	R_t^*= -\xi_T +\int_t^T\Bigl(-h^*_s\sigma(V_s)S_s\lambda_s+ g\bigl(Z_{1,s}^*-h^*_s\sigma(V_s)S_s, Z_{2,s}^*\bigr)\Bigr) \, ds -\int_t^T Z^*_{1,s} \, dW_{1,s} -\int_t^T Z^*_{2,s} \, dW_{2,s}.
 \]

We state now two useful results about uniform bounds on $R^m$ and $\bZ^m := \bigl(Z_1^m,Z_2^m\bigr)$ which will be proved later:
\begin{lemma}\label{lem:Appendix-uniform_Rm_bound}
 	There exist uniform bounds $L, U\in \mathbb R$ such that for every $m$, 
 	\[
 		L\leq R^m_t \leq U \quad a.s. \qquad \text{ for all }t\in [0, T].
 	\]
\end{lemma}

\begin{lemma}\label{lem:Appendix-uniform_Zm_bound}
	There exists a uniform bound $K >0$, such that
 	\[
 		\E\biggl[\int_0^T\bigl(Z^m_{1,t}\bigr)^2 \, dt +\int_0^T\bigl(Z^m_{2,t}\bigr)^2 \, dt\biggr]< K \qquad \text{for all } m.
	\]
\end{lemma}
 
Define $\bar{R}^m_s:=R_s^m-R_s^*$,  $\bar{g}_s:=  g\bigl(Z_{1,s}^m-h^m_s\sigma(V_s)S_s, Z_{2,s}^m\bigr)-  g\bigl(Z_{1,s}^*-h^*_s\sigma(V_s)S_s, Z_{2,s}^*\bigr),$
$\bar{Z}^m_{1,s}:= Z^m_{1,s} - Z^*_{1,s}$,  $\bar{Z}^m_{2,s}:= Z^m_{2,s} - Z^*_{2,s}$ and $\bar{h}^m:=h^m-h^*$. For simplicity write $\Zm:=\bigl(\bar{Z}^m_1, \bar{Z}^m_2\bigr)$ and $\W:=\bigl(W_1, W_2\bigr)$.  Observe that $\bar{R}^m_t\geq 0$ for $t\in [0,T]$, $\bar{R}^m_T=0$  and
\[
	-d\bRm_s= \big(-\hm\sigma(V_s)S_s\lambda_s+ \bar{g}_s\bigr) \, ds -\Zm_s\cdot \,  d\W_s .
\]

For $\alpha>0$, the value of $\alpha$ to be determined later, we apply  It\^o's formula to $e^{\alpha \bar{R}^m}$ which yield
\[
	e^{\alpha \bRm_t} = 1-\alpha \int_t^T e^{\alpha \bRm_s}\hm_s\sigma(V_s)S_s\lambda_s \, ds +\alpha \int_t^Te^{\alpha \bRm_s} \bar{g}_s \, ds 
 	+\alpha \int_t^Te^{\alpha \bRm_s}\Zm_s\cdot \, d\W_s -\frac{\alpha^2}{2}\int_t^T e^{\alpha \bRm_s}  |\Zm_s|^2 \, ds .
 \]
 
Taking the conditional expectation with respect to $\mathcal F_t$, we get
\begin{align}\label{eq:Ito-formula}
	e^{\alpha \bRm_t} & = 1  - \alpha \E_t\biggl[\int_t^T e^{\alpha \bRm_s}\hm_s\sigma(V_s)S_s\lambda_s \, ds\biggr]+\alpha \E_t\biggl[\int_t^Te^{\alpha \bRm_s} \bar{g}_sds \biggr]
 	-\frac{\alpha^2}{2}\E_t\biggl[\int_t^T e^{\alpha \bRm_s}  |\Zm_s|^2ds\biggr] \nonumber\\
 	&= 1 - \alpha \E_t\biggl[\int_t^T e^{\alpha \bRm_s}\hm_s\sigma(V_s)S_s\lambda_s \, ds\biggr] +I_1+I_2.
 \end{align}

We first bound $I_1$ as follows.  Using the bounds on the partial derivatives of $g$, we get
\[
	I_1  = \alpha \E_t\biggl[\int_t^Te^{\alpha \bRm_s} \bar{g}_s \, ds \biggr] \leq  \alpha C_1 \E_t\biggl[\int_t^Te^{\alpha \bRm_s}\Bigl(1+\bigl\vert \bZ^m \bigr\vert+\bigl\vert \bZ^*_s \vert+\vert h^*S_s \bigr\vert \Bigr) \Bigl(\bigl\vert \Zm \bigr\vert +\bigl\vert \hm  S_s \bigr\vert \Bigr) \, ds \biggr] 
\]
using Cauchy-Schwarz we get
\[
	I_1  \leq  \alpha C_2 \biggl(\int_t^T \E_t\Bigl[e^{\alpha \bRm_s}\Bigl(1+\bigl\vert \bZ^m\bigr\vert^2+\bigl\vert \bZ^*_s\bigr\vert^2+\bigl\vert h^*S_s\bigr\vert^2\Bigr)\Bigr] \, ds \int_t^T \E_t\Bigl[e^{\alpha \bRm_s}\Bigl(\bigl\vert \Zm\bigr\vert ^2+\bigl\vert \hm S_s\bigr\vert ^2\Bigr)\Bigr] \, ds\biggr)^{1/2}
	\]
and using \eqref{eq:h*bound} and  Lemmas \ref{lem:Appendix-uniform_Rm_bound} and  \ref{lem:Appendix-uniform_Zm_bound}, we get
\[
	I_1   \leq  \alpha C_3 \Biggl[\sqrt{\int_t^T \E_t \Bigl[ e^{\alpha \bRm_s} \bigl\vert \Zm \bigr\vert^2\bigr] \, ds}+\sqrt{\int_t^T \E_t \Bigl[e^{\alpha \bRm_s} \bigl\vert \hm S_s\bigr\vert ^2\Bigr] \, ds}\Biggr].
\]

Let $A_m:=\int_t^T \E_t \bigl[e^{\alpha \bRm_s}\bigl\vert \Zm \bigr\vert^2 \bigr] \, ds$, then
\[
	I_1\leq \alpha C_3 \sqrt{A_m} +\alpha C_3 \sqrt{\int_t^T\E_t \bigl[e^{\alpha \bRm_s}\bigl\vert \hm S_s\bigr\vert ^2\bigr] \, ds},
\]
where the constant $C_3$ is independent of $m$.

Observe that $I_2= -\frac{\alpha^2}{2}A_m$. Let $A_\infty:=\limsup_{m\to\infty} A_m$. Suppose $A_\infty>0$. Let us consider a subsequence $\{m_k\}$ along which $A_{m_k}$ converges to $A_\infty$. Then for large $k$,  $A_{m_k}> A_\infty/2$. Choosing $\alpha = 4C_3 \sqrt{2/A_\infty}$ we get
 
\begin{equation}\label{eq:Bound_I1+I2}
	I_1+I_2\leq \alpha C_3 \sqrt{\int_t^T \E_t \Bigl[e^{\alpha \bar{R}^{m_k}_s} \bigl \vert \bar{h}^{m_k} S_s\bigr\vert^2\Bigr] \, ds} - \frac{\alpha^2}{4} A_{m_k}\qquad \text{ for large }k.
\end{equation}
Note that $\alpha$ does not depend on the index $m_k$.  

Along this subsequence, we can bound $e^{\alpha \bar{R}^{m_k}_t}$ in \eqref{eq:Ito-formula} as follows:
\[
	e^{\alpha \bar{R}^{m_k}_t}\leq 1 - \alpha \E_t\biggl[\int_t^T e^{\alpha \bar{R}^{m_k}_s}\bar{h}^{m_k}_s\sigma(V_s)S_s\lambda_s \, ds\biggr]+\alpha C_3 \sqrt{\int_t^T \E_t \Bigl[e^{\alpha \bar{R}^{m_k}_s}\bigl\vert\bar{h}^{m_k} S_s\bigr\vert^2\Bigr] \,ds} - \frac{\alpha^2}{4} A_{m_k}.
\]
We rewrite this  as
 \[
 	\frac{\alpha^2}{4} A_{m_k}+e^{\alpha \bar{R}^{m_k}_t}\leq 1 - \alpha \E_t\biggl[\int_t^T e^{\alpha \bar{R}^{m_k}_s}\bar{h}^{m_k}_s\sigma(V_s)S_s\lambda_s \, ds\biggr]+\alpha C_3 \sqrt{\int_t^T\E_t e^{\alpha \bar{R}^{m_k}_s} \bigl\vert \bar{h}^{m_k} S_s \bigr\vert^2)ds}.
 \]
 Taking $\limsup_{k\to\infty}$ and using $\bar{h}^{m_k}\to0$, by the dominated convergence theorem we get
 \[
 	\frac{\alpha^2}{4} A_\infty+\limsup_{k\to\infty}e^{\alpha \bar{R}^{m_k}_t} \leq 1.
\]
Since $e^{\alpha \bar{R}^{m_k}_t}\geq 1$ and $A_{m_k}\geq 0$, this gives us $A_\infty=0$. This implies $A_m\to 0$ along the original full sequence. Reverting to the full sequence and taking $\limsup_{m\to\infty}$ of \eqref{eq:Ito-formula} we get
\[
	\limsup_{m\to\infty} e^{\alpha \bRm_t} \leq 1
\] and conclude that $\lim_{m\to \infty} R^m = R^*$.
\end{proof}

\begin{proof}[Proof of Lemma \ref{lem:Appendix-uniform_Rm_bound}]
Recall that  $- g^*(-\lambda_s, z_2)-\lambda_s z_1 = \inf_{\theta\in \R}\bigl\{-\theta\lambda_s+g( z_1-\theta, z_2)\bigr\}$. Taking $\theta= h^m_s\sigma(V_s)S_s$, we see that
\[
	-\lambda_s z_1-g^*_s \bigl(-\lambda_s, z_2\bigr) \leq g^m_s\bigl(z_1, z_2\bigr), \text{ for all } s\in [0,T],
\]
where $g^m_s(z_1, z_2):= -h^m_s\sigma(V_s)S_s\lambda_s+ g\bigl(z_1-h^m_s\sigma(V_s)S_s, z_2\bigr)$. By the comparison theorem for quadratic BSDEs given in \cite[Theorem 5]{BH08}, 
$R^*_t \leq R^m_t$ for all $t$.  

On the other hand, we can bound $g^m(z_1,z_2) \leq C\bigl(1+z_1^2+z_2^2\bigr) = : g^0(z_1, z_2)$ using the quadratic growth of the driver $g$ and the bounds on $h^*S$ in \eqref{eq:h*bound}. Thus, by the comparison principle, $R_t^m\leq R_t^0$ for all $t$, where $(R^0, Z_1^0, Z_2^0)$ is a solution to the BSDE
\[
	R_t^0= -\xi_T+\int_t^Tg^0\bigl(Z_{1,s}^0, Z_{2,s}^0\bigr) \, ds - \int_t^T Z^0_{1,s} \, dW_{1,s} - \int_t^T Z^0_{2,s} \, dW_{2,s},
\]
with bounded $R^0$.

Since $R^*$ and $R^0$ are bounded and $R_t^*\leq R^m_t\leq R^0_t$ a.s. for all $t$, we get the result.
\end{proof}

\begin{proof}[Proof of Lemma \ref{lem:Appendix-uniform_Zm_bound}]
Let $\alpha>0$ be a constant whose value will be determined later. Recall that
\[
	-dR^m_t= \Bigl(-h^m_s\sigma(V_s)S_s\lambda_s+ g\bigl(Z_{1,s}^m-h^m_s\sigma(V_s)S_s, Z_{2,s}^m\bigr)\Bigr) \, ds -\int_t^T Z^m_{1,s} \, dW_{1,s} -\int_t^T Z^m_{2,s}\, dW_{2,s}.
\]
Applying It\^o's formula and taking expectations, we get
\begin{align*}
	\E_t\bigl[e^{\alpha R_t^m}\bigr] &= \E_t \bigl[e^{\alpha \xi_T}\bigr]+\alpha \E_t \biggl[\int_t^Te^{\alpha R^m_s}\Bigl(-h^m_s\sigma(V_s)S_s\lambda_s+ g\bigl(Z_{1,s}^m-h^m_s\sigma(V_s)S_s, Z_{2,s}^m\bigr)\Bigr) \, ds\biggr]\\
	&\phantom{=:} -\frac{\alpha^2}{2} \E_t \biggl[\int_t^T e^{\alpha R^m_s} \bigl(Z^m_{1,s}\bigr)^2 \, ds +\int_t^T e^{\alpha R^m_s} \bigl(Z^m_{2,s}\bigr)^2 \, ds\biggr] \\
	&\leq  \E_t \bigl[e^{\alpha \xi_T}\bigr]+\alpha C \E_t \biggl[ \int_t^Te^{\alpha R^m_s}\Bigl(1+ (h^*_sS_s)^2 +\bigl(Z^m_{1,s}\bigr)^2+\bigl(Z^m_{2,s}\bigr)^2\Bigr) \, ds \biggr]\\
	&\phantom{=:} -\frac{\alpha^2}{2} \E_t\biggl[\int_t^T e^{\alpha R^m_s} \bigl(Z^m_{1,s}\bigr)^2 \, ds +\int_t^T e^{\alpha R^m_s} \bigl(Z^m_{2,s}\bigr)^2 \, ds\biggr] 
\end{align*}
by the quadratic growth of $g$ and the boundedness assumptions on the model coefficients. Taking $\alpha>4C$ in the above, we get
\[
	\E_t\bigl[e^{\alpha R_t^m}\bigr] \leq  \E_t \bigl[e^{\alpha \xi_T}\bigr]+\alpha C \E_t \biggl[\int_t^Te^{\alpha R^m_s}\Bigl(1+ (h_s^*S_s)^2\Bigr) \, ds\biggr] - \frac{\alpha^2}{4}\E_t \biggl[\int_t^T e^{\alpha R^m_s} \bigl(Z^m_{1,s}\bigr)^2 \, ds +\int_t^T e^{\alpha R^m_s} \bigl(Z^m_{2,s}\bigr)^2 \, ds\biggr].
\]
Since $\E_t\bigl[e^{\alpha R_t^m}\bigr] \geq 0$, we have
\[
	\E_t\biggl[\int_t^T e^{\alpha R^m_s} \bigl(Z^m_{1,s}\bigr)^2ds +\int_t^T e^{\alpha R^m_s} \bigl(Z^m_{2,s}\bigr)^2 \, ds\biggr] \leq \frac{4}{\alpha^2} \E_t \bigl[e^{\alpha \xi_T}\bigr] + \frac{4C}{\alpha} \E_t \biggl[ \int_t^Te^{\alpha R^m_s}\bigl(1+ S_s^2\bigr) \, ds \biggr].
\]
Using the bounds in Lemma \ref{lem:Appendix-uniform_Rm_bound} on $R^m$, 
\[
	e^{\alpha L} \E_t\biggl[\int_t^T  \bigl(Z^m_{1,s}\bigr)^2ds +\int_t^T e^{\alpha R^m_s} \bigl(Z^m_{2,s}\bigr)^2 \, ds\biggr] \leq \frac{4}{\alpha^2} \E_t \bigl[e^{\alpha \Vert \xi_T\Vert_\infty}\bigr] + \frac{4C}{\alpha} e^{\alpha U}\int_t^T\Bigl(1+ \E_t\bigl[S_s^2\bigr]\Bigr) \, ds.
\]
The boundedness of $\int_t^T \E_t \bigl[S_s^2\bigr] \, ds$ implies the result.
\end{proof}

\subsection{The RBSDE case}

Let $\bigl(\HR^*, \HZ_1^*, \HZ_2^*, \HK^*\bigr)$ be the solution to the  RBSDE with terminal condition $\xi_T$, driver $- g^*(-\lambda_s, z_2)-\lambda_s z_1 = \inf_{\theta\in \R}\bigl\{-\theta\lambda_s+g( z_1-\theta, z_2)\bigr\}$ and boundary $\xi_t$. Let $\bigl(\HR^h, \HZ_1^h, \HZ_2^h, \HK^h\bigr)$ denote the solution to the RBSDE with terminal condition $\xi_T$, driver $-h_s\sigma(V_s)S_s\lambda_s+ g(z_1- h_s\sigma(V_s)S_s, z_2)$ and boundary $\xi_t$. Note that the only difference between the two RBSDEs is in the driver. By the comparison principle for RBSDEs, see \cite{KLQT02}, we have $\HR_t^*\leq \HR_t^h$ for any choice of strategy $h\in \mathcal H$ and hence $R_t^*\leq \inf_{h\in \mathcal H}R_t^h$. To prove equality, we show that $\HR_t^*=\lim_{m\to\infty}\HR_t^{h_m}$ a.s. for some sequence $h^m\in \mathcal H$.

We choose $h^m$ as we did above in the BSDE case and the argument follows similarly. The main difference lies in the additional term coming from $K$.  Thus

\[
	-d\bar{R}^m_s= \bigl(-\hm_s \sigma(V_s)S_s\lambda_s+ \bar{g}_s)\bigr) \, ds - \bar{Z}^m_{1,s} \, dW_{1,s} - \bar{Z}^m_{2,s} \, dW_{2,s} + d\bar{K}^m_s,
\]
where $\bar{R}^m_s:=\HR_s^m-\HR_s^*$, $\bar{g}_s:=  g\bigl(\HZ_{1,s}^m-h^m_s\sigma(V_s)S_s, \HZ_{2,s}^m\bigr)-  g\bigl(\HZ_{1,s}^*-h^*_s\sigma(V_s)S_s, \HZ_{2,s}^*\bigr),$
$\bar{Z}^m_{1,s}:= \HZ^m_{1,s} - \HZ^*_{1,s}$,  $\bar{Z}^m_{2,s}:= \HZ^m_{2,s} - \HZ^*_{2,s}$, $\bar{h}^m:=h^m-h^*$ and  $\bar{K}_s:= \HK^m_s-\HK^*_s$. Then 
\begin{align*}
	e^{\alpha \bar{R}^m_t} = & 1 -\alpha \int_t^T e^{\alpha \bar{R}^m_s} \hm_s\sigma(V_s)S_s\lambda_s \, ds +\alpha \int_t^T e^{\alpha \bar{R}^m_s}  \bar{g}_s \, ds -\alpha \int_t^T\bar{R}^m_s\bar{Z}^m_{1,s} \, dW_{1,s} \\
	&-\alpha\int_t^T \bar{R}^m_s \bar{Z}^m_{2,s} \, dW_{2,s} -\frac{\alpha^2}{2}\int_t^T e^{\alpha \bar{R}^m_s} \bigl(\bar{Z}^m_{1,s}\bigr)^2 \, ds-\frac{\alpha^2}{2}\int_t^T e^{\alpha \bar{R}^m_s}\bigl(\bar{Z}^m_{2,s}\bigr)^2 \, ds +\alpha\int_t^T e^{\alpha \bar{R}^m_s} \, d\bar{K^m_s}.
\end{align*}

By Proposition 5 in \cite{L14}, we have $d\bar{K}^m_s \leq 0$, and since $e^{\alpha \bar{R}^m_s} \geq 0$, we have
\begin{align*}
	e^{\alpha \bar{R}^m_t} & \leq  1 -\alpha \int_t^T e^{\alpha \bar{R}^m_s} \hm_s\sigma(V_s)S_s\lambda_s \, ds +\alpha \int_t^T e^{\alpha \bar{R}^m_s}  \bar{g}_s \, ds -\alpha \int_t^T\bar{R}^m_s\bar{Z}^m_{1,s} \, dW_{1,s} \\
	& \phantom{=:} -\alpha\int_t^T \bar{R}^m_s \bar{Z}^m_{2,s} \, dW_{2,s} -\frac{\alpha^2}{2}\int_t^T e^{\alpha \bar{R}^m_s} \bigl(\bar{Z}^m_{1,s}\bigr)^2 \, ds -\frac{\alpha^2}{2}\int_t^T e^{\alpha \bar{R}^m_s}\bigl(\bar{Z}^m_{2,s}\bigr)^2 \, ds.
\end{align*}
We then proceed as in the BSDE case to get $\HR_t^*=\lim_{m\to\infty} \HR^m_t$.
\end{document}